\documentclass[12pt]{article}
\pdfoutput=1

\usepackage{amsmath,amsbsy,amssymb,bm}
\usepackage{authblk}
\usepackage{amsthm}
\usepackage{graphicx,epstopdf}
\usepackage{bm}
\usepackage{color}
\usepackage{natbib}
\usepackage[font=footnotesize]{caption}
\usepackage{subcaption} 

\theoremstyle{definition}
\newtheorem{condition}{Condition}[section]
\theoremstyle{remark}
\newtheorem{remark}{Remark}[section]
\theoremstyle{plain}
\newtheorem{proposition}{Proposition}[section]
\newtheorem{lemma}{Lemma}[section]

\let\oldtabular\tabular
\renewcommand{\tabular}{\footnotesize\oldtabular}

\title{\bf Multivariate Geometric Anisotropic \\Cox Processes}
\author[1]{J. S. Martin}
\author[2]{D. J. Murrell}
\author[3]{S. C. Olhede}
\affil[1]{{\small Department of Mathematics, Imperial College London, London SW11 2AZ, U.K.}}
\affil[2]{{\small Centre for Biodiversity and Environment Research, Department of Genetics, Evolution and Environment, University College London, London WC1E 6BT, U.K.}}
\affil[3]{{\small Department of Statistical Science, University College London, London WC1E 6BT, U.K.}}

\date{\today}
\advance\day by -1

\begin{document}
\def\blind{0} 
\def\blindtest{1}

\def\spacingset#1{\renewcommand{\baselinestretch}%
	{#1}\small\normalsize} \spacingset{1}



\ifx\blind\blindtest
\bigskip
\bigskip
\bigskip
\begin{center}
	{\LARGE\bf Multivariate Geometric Anisotropic Cox Processes}
\end{center}
\medskip
\else
\maketitle
\fi

\bigskip
\begin{abstract}
	This paper introduces a new modelling framework for multivariate anisotropic Cox processes. Building on recent innovations in multivariate spatial statistics, we propose a new family of multivariate anisotropic random fields and construct a family of anisotropic point processes from it. We give conditions that make the models valid, and we provide additional understanding of valid point process dependence. We also propose a likelihood-based inference mechanism for this type of process. Finally we illustrate the utility of the proposed modelling framework by analysing spatial ecological observations of plants and trees in the Barro Colorado Island study.
\end{abstract}

\noindent%
\vfill

\newpage
\spacingset{1.45} 
\section{Introduction}
\label{sec:intro}

	In this paper, we introduce a new class of multivariate and heterogeneous point process models. In doing so, we address two challenging problems in point process analysis: we propose valid and nontrivial models for multi-type point processes, an open problem in the literature, and we produce multivariate spatial models that can flexibly accommodate anisotropy in both the marginal and joint dependence structures. 

	We choose to build our models using log-Gaussian Cox processes \citep{Diggle83a,Moller98} as a foundation. Thus, the observed point pattern is modelled in terms of a random intensity, generated by a random field. Recent interest in random field modelling has greatly enhanced our ability to specify flexible models for multivariate patterns. We shall build on recent progress made in this area by, for example, \citet{Gneiting10}, \citet{Apanasovich12} and \citet{Genton15}, by allowing for anisotropy in the second-order dependence structure of the latent random field model.

	Datasets that require anisotropic models have been common in the point process literature over the last 20 years, for example the locations of chapels in Welsh valleys \citep{Mugglestone96a,Rajala16,Rajala18a}, the epicentral locations of earthquakes in California over a 20 year period \citep{Veen06} and clustered locations of shrubs in dryland ecosystems \citep{Haase01}. The Welsh chapels and Californian earthquakes both form elliptical clusters, indicating an anisotropic second-order interaction between points in the same pattern. Meanwhile, the dryland shrub data display a directional preference in the interaction of points of different type: \citet{Haase01} found one species to grow more often than would be expected to the east of a second species. In the point process literature, it is common to accommodate heterogeneities in the observed point pattern by using a spatially homogeneous random field model to specify an intensity process conditional upon some known covariates \citep[see, e.g.~][]{Waagepetersen08,Waagepetersen09,Diggle13b}. This approach is limited in its applicability, however, when faced with heterogeneous point pattern data with no covariate measurements, or indeed when the source of heterogeneity is unknown.

	Our chosen approach to accommodating anisotropy is based upon a particular form of anisotropy known as geometric anisotropy \citep{Goff88}: whereas the spatial covariance functions that drive isotropic processes have circular contours of equivariance, those that drive geometric anisotropic processes have elliptical contours of equivariance. This approach was also considered in the univariate case by \citet{Moller14}. A great advantage of this approach is that it can be used in conjunction with well-known isotropic covariance functions; our models will use Mat\'ern-based covariance structures, which will allow the user to directly specify both the range of dependence in, and the smoothness of, the resulting random field. We will also discuss and address identifiability concerns for this class of parametric models.

	Once the random field has been specified, the point process is conditionally generated as a Poisson process with intensity determined by the random field. This has the advantage of automatically generating a valid set of anisotropic point processes, where the marginal and cross-pair correlation functions have an analytic form, which we provide. 
	We explore the restrictions that are naturally placed on all cross-pair correlation functions, where we utilise recent results for isotropic multivariate random fields due to \citet{Apanasovich12} and \citet{Gneiting10}. 
	By representing our multivariate process in both the spatial and spectral domains, we will also demonstrate that allowing for distinct geometric anisotropies in each marginal process places further restrictions on valid forms of the cross-dependence structures. This is an important result that yields unique insights into the possible variation of joint co-dependence in multivariate geometric anisotropic random fields, and by extension Cox processes. 


	Once we have understood the constraints on possible model forms, we develop new inference methods. We detail a two-stage estimation procedure in which we first estimate the anisotropy parameters, and then use these estimates to transform the data to be isotropic; this `isotropised' point pattern is then used to estimate the covariance parameters for the underlying random field model. For this second stage, \citet{Moller14} advocated the use of minimum contrast, a method of moments approach to estimation for point process models. We appeal to the likelihood principle, and develop a maximum likelihood-based approach to inference that builds on the work of \citet{Tanaka08}. Straightforward maximum likelihood estimation of the model parameters is infeasible, due to the intractability of the LGCP likelihood, however \citet{Tanaka08} showed that the intractability of the point process likelihood can be circumvented by considering the so-called Fry process \citep{Fry79}. This is a secondary point pattern formed by the difference vectors of all point pairs in the original point pattern, and it can be treated as an inhomogeneous Poisson point process, with an associated tractable likelihood. \citet{Tanaka08} showed that the Fry process likelihood can be used to perform inference for univariate, isotropic point process models. In a novel extension of this work, we use the Fry process likelihood to perform inference for anisotropic, multivariate point processes.

	Finally, we apply our newly-developed methodology to real data from a tropical rainforest stand on Barro Colorado Island, Panama \citep{Condit98,Hubbell99,Hubbell10}. Recent work by \citet{Waagepetersen16} and \citet{Rajala18b} has highlighted the importance of developing realistic multivariate point process models to aid the understanding of complex species interactions within this rainforest. The need to develop anisotropic methodology in particular is characterised in Figures \ref{subfig:BCIdata_int1} and \ref{subfig:BCIdata_int2}, which show the estimated intensity of \textit{Guatteria dumetorum} and \textit{Miconia hondurensis}. Their strongly anisotropic features are clear, and we also show two simulated fields from the presented multivariate model class, exhibiting similar features. 

	\begin{figure}[t!]
	\hspace{0.01\textwidth}
		\begin{subfigure}{0.4\textwidth}
			\includegraphics[width=\textwidth]{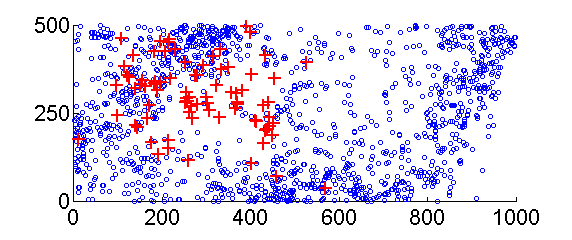}
			\vspace{-10pt}
			\vspace{-10pt}
			\caption{}
			\label{subfig:BCIdata_PPs}
		\end{subfigure}	
	\hspace{0.1\textwidth}
		\begin{subfigure}{0.4\textwidth}
			\includegraphics[width=\textwidth]{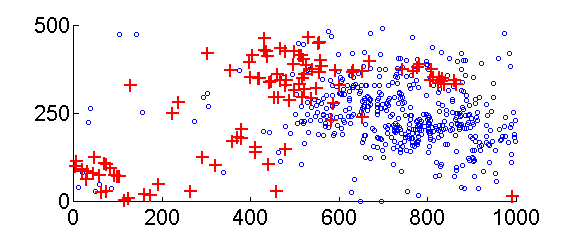}
			\vspace{-10pt}
			\vspace{-10pt}
			\caption{}
			\label{subfig:synthdata_PPs}
		\end{subfigure}
	\begin{center}
	\vspace{-10pt}
		\begin{subfigure}{0.475\textwidth}
			\includegraphics[width=\textwidth]{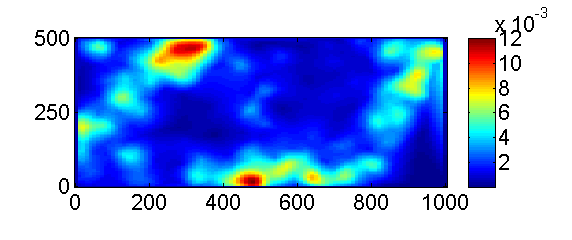}
			\vspace{-10pt}
			\vspace{-10pt}
			\caption{\hspace{15pt} }
			\label{subfig:BCIdata_int1}
		\end{subfigure}
	\hspace{0.03\textwidth}
		\begin{subfigure}{0.475\textwidth}
			\includegraphics[width=\textwidth]{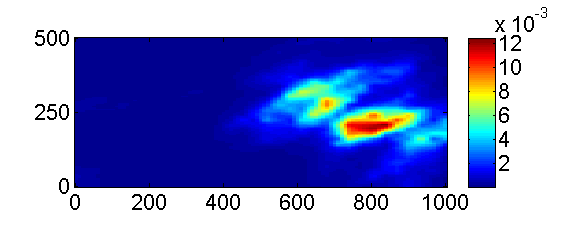}
			\vspace{-10pt}
			\vspace{-10pt}
			\caption{\hspace{15pt} }
			\label{subfig:synthdata_int1}
		\end{subfigure}
	\vspace{-20pt}
		\begin{subfigure}{0.475\textwidth}
			\includegraphics[width=\textwidth]{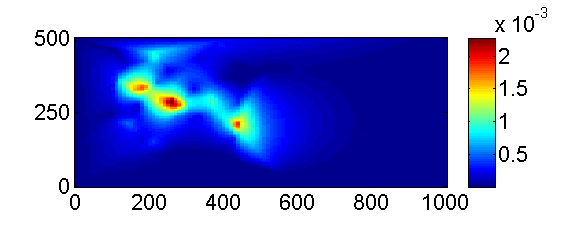}
			\vspace{-10pt}
			\vspace{-10pt}
			\caption{\hspace{15pt} }
			\label{subfig:BCIdata_int2}`
		\end{subfigure}
	\hspace{0.03\textwidth}
		\begin{subfigure}{0.475\textwidth}
			\includegraphics[width=\textwidth]{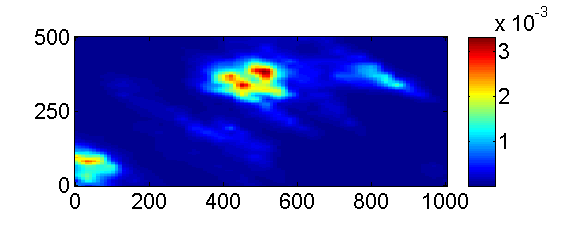}
			\vspace{-10pt}
			\vspace{-10pt}
			\caption{\hspace{15pt} }
			\label{subfig:synthdata_int2}
		\end{subfigure}
	\end{center}
	\caption{Point pattern data showing two species of tree from the BCI tropical rainforest (a; {\em Guatteria dumetorum}, blue; {\em Miconia hondurensis}, red), along with their estimated intensity fields (c,e), and simulated point pattern data (b) from two independent univariate geometric anisotropic log-Gaussian Cox processes, with their corresponding simulated intensity fields (d,f).}
	\label{fig:BCIvssynthdata}
	\end{figure}

	Thus, to summarize, this paper provides a number of new and important insights for multivariate spatial processes, describing the complex relationships possible when allowing for distinct geometric anisotropies in each univariate component. Our understanding gives sufficient, but not necessary, conditions to yield valid multivariate random field models and, by extension, valid multivariate Cox processes. 

	\section{Background}
	\subsection{Log-Gaussian Cox processes}
	Consider the multivariate point process $X=\{X_p\in\mathbb{R}^d, p=1,\ldots,P\}$, where the index $p$ is used to denote a univariate component of the multivariate process, and suppose that we wish to use such a process to model a multi-type point pattern. We will denote the observed point pattern $X\cap W = \{x_{p,i} \in W, i=1,\ldots,n_p ; p=1,\ldots,P\}$, where $n_p\in\mathbb{N}$ is the total number of points of type $p$ observed in the observation window $W\subset\mathbb{R}^d$. Henceforth, we will also use $x_p$ to denote an arbitrary observed point of type $p$. For many applications of interest, $d=2$, however much of the multivariate framework established here can be applied to point processes defined on a space of any dimension. 

	We define $X$ to be a multivariate log-Gaussian Cox process \citep[LGCP;][]{Moller98}: each univariate sub-process $X_p$ is an inhomogeneous Poisson process with intensity specified by
	\begin{equation}
	\label{eq_LGCPintensitydef}
	\Lambda_p(x) = \exp\{S_p(x)\}, \qquad x\in\mathbb{R}^d,
	\end{equation}
	where $S(x) = \{S_p(x), p=1,\ldots,P\}$ is a multivariate Gaussian random field (GRF). We will assume $S_p$, and therefore $X_p$, to be stationary for all $p=1,\ldots,P$, and we denote the constant mean of $S_p(x)$ by $\mu_p$. The intensity process $\Lambda_p(x)$ will therefore also have a constant mean, which we denote $\lambda_p$, and which will take the following form:
	\begin{equation}
	\label{eq:lamp_mup}
	\lambda_p = \mathbb{E}\left\{\Lambda_p\right\} = \exp\left\{\mu_p + \sigma_{pp}/2\right\},
	\end{equation}
	where $\sigma_{pp}$ denotes the variance of $S_p(x)$.

	Key to the definition of a multivariate LGCP is the conditional independence of its components: given its intensity process $\Lambda_p(x)$, the univariate LGCP $X_p$ is independent of $\{X_q,q=1,\ldots,P, q\ne p\}$. As a result of this property, the second-order behaviour of the point process $X$ may be entirely, and conveniently, described through the covariance structure of the multivariate GRF $S(x)$. We do so by specifying the matrix of covariance functions $\{C_{pq}(h)\}_{p,q=1}^P$, with
	\begin{equation*}
		C_{pq}(h) = \textrm{cov}\left\{S_{p}(x),S_{q}(x-h)\right\},\qquad x,h\in\mathbb{R}^d.
	\end{equation*}

	The second-order behaviour of the multivariate point process $X$ can be directly measured through the level of clustering or separation present in the resulting point pattern. The cross-pair correlation function $g_{pq}(r)$ is defined as the expected number of points from process $q$ that lie at a distance $r$ from the typical point in process $p$, and is the standard tool for measuring aggregation and segregation, both within and between processes. For a log-Gaussian Cox process, $g_{pq}(r)$ can be straightforwardly expressed in terms of the covariance structure for the underlying multivariate GRF:
	\begin{equation*}
	g_{pq}(h) = \exp\{C_{pq}(h)\}, \qquad h\in\mathbb{R}^d.
	\end{equation*}
	From this relationship, it is clear to see that $g_{pq}(h)=1$ is equivalent to $C_{pq}(h)=0$, which indicates independence between processes $p$ and $q$ at the scale $h\in\mathbb{R}^d$. Thus, for a bivariate Poisson process $\{X_p,X_q\}$, i.e.~under an assumption of complete spatial randomness, we would expect $g_{pq}(h)=1$, whereas significant departures from this would indicate aggregation ($g_{pq}(h)>1$) or segregation ($g_{pq}(h)<1$) of points from processes $p$ and $q$ at separation $h\in\mathbb{R}^d$.

	The dependence structure for the multivariate GRF $S(x)$ can equivalently be described in the frequency (spectral) domain. The (cross-)spectral density function $f_{pq}(\omega)$ forms a Fourier transform pair with the (cross-)covariance function:
	\begin{equation*}
	f_{pq}(\omega) = \frac{1}{(2\pi)^d}\int_{\mathbb{R}^d} \exp(-i\omega^T x) C_{pq}(x)dx, \qquad \omega\in\mathbb{R}^d.\\
	\end{equation*}
	By considering the spectral-domain behaviour of our multivariate GRF $S(x)$, we will demonstrate the difficulties inherent in multivariate modelling of geometric anisotropic spatial dependence, and we will consider the complex coherence at frequency $\omega\in\mathbb{R}^d$, $\gamma_{pq}(\omega)$:
	\begin{equation}
	\label{eq:compcoherence}
	\gamma_{pq}(\omega) = \frac{f_{pq}(\omega)}{\left\{f_{pp}(\omega)f_{qq}(\omega)\right\}^{\frac12}}.
	\end{equation}

	\subsection{Geometric Anisotropic LGCPs}
	\label{subsec:GALGCPs}
	We describe here the approach to modelling geometric anisotropy in univariate LGCPs, as introduced by \citet{Moller14}. In brief, the required dependence structure is specified through the application of an isotropic covariance structure to a geometrically manipulated version of the space on which the process lives. Since a LGCP is fully defined by the first and second order characteristics of the underlying Gaussian random field, \citet{Moller14} showed that one can therefore construct a geometric anisotropic LGCP through using standard geometric manipulations to modify the space on which the latent univariate GRF is defined. In Section \ref{sec:mvga}, we will show how this can be flexibly extended to the multivariate case by specifying individual components of a population of $P$ GRFs through $P$ potentially distinct geometric manipulations of $\mathbb{R}^2$.

	Given an isotropic covariance function $C_0(\|h\|)$, $h\in\mathbb{R}^d$, we can define a geometric anisotropic version as \citep{Christakos92}
	\begin{equation}
	\label{eq:geoanisotcov}
	C(h) = C_0\left(\sqrt{h^T\Sigma^{-1} h}\right),
	\hspace{25pt}
	h\in\mathbb{R}^d,
	\end{equation}
	where
	\begin{equation}
	\label{eq:geoanisotSigmadef}
	\Sigma = R_{\theta}
	\left(
	\begin{array}{cc}
	1	&	0	\\
	0	&	\zeta^2
	\end{array}
	\right)
	R^T_{\theta},
	\end{equation}
	for $\theta\in[0,\pi)$ and $\zeta\in(0,1]$, and where $R_{\theta}$ is the rotation matrix. Under this parameterisation, $\Sigma$ is defined such that the ellipse $E=\{h\in\mathbb{R}^2:h^T\Sigma^{-1}h=1\}$ has a semi-major axis of unit length at angle $\theta$, relative to the abscissa axis of the original coordinate system, and a semi-minor axis of length $\zeta$ at angle $\theta+\pi/2$. Accordingly, we can describe the covariance function defined in \eqref{eq:geoanisotcov} as `elliptic', and we have that the LGCP driven by a GRF with elliptic covariance structure will also display elliptic, or geometric anisotropic, second-order behaviour, described by the pair correlation function and spectral density function as:
	\begin{eqnarray*}
	g(h) & = & g_0\left(\sqrt{h^T\Sigma^{-1} h}\right) = \exp\left\{C_0\left(\sqrt{h^T\Sigma^{-1} h}\right)\right\} \\
	f(\omega) & = & \left|\Sigma\right|^{1/2} f_0\left(\sqrt{\omega^T\Sigma\omega}\right),
	\end{eqnarray*}
	for $h,\omega\in\mathbb{R}^d$, where $f_0(\|\omega\|)$ is the isotropic spectral density that forms a Fourier transform pair with $C_0(\|h\|)$, and $g_0(\|h\|)$ is the corresponding isotropic pair correlation function.

	Our specification of geometric anisotropy differs slightly from that of \citet{Moller14}, who include an additional scale parameter in their definition of the deformation matrix $\Sigma$; this is used to scale the axes in the resulting elliptical covariance structure. In practice, however, the majority of parametric covariance functions of interest incorporate a scale parameter that directly controls the correlation length, and so including a separate scale parameter in \eqref{eq:geoanisotSigmadef} creates nonidentifiability issues when performing parameter inference. We avoid this issue by assuming all scale information to be controlled by the parametric form of $C_0(\|h\|)$.

	Since we are considering processes that display anisotropy, it will be useful for their analysis to be able to express their second-order properties in polar coordinates. We therefore define the anisotropic pair correlation function, replacing the vector $h\in\mathbb{R}^d$ with its length $r$ and angle $\phi$:
	\begin{equation}
	\label{eq:anisotpcf}
	g^a(r,\phi) = g([r\cos\phi,r\sin\phi]) = g_0\left(\frac{r}{\zeta}\sqrt{1-(1-\zeta^2)\cos^2(\phi-\theta)}\right).
	\end{equation}


	\section{Defining the Model}
	\label{sec:model}
	\subsection{Accommodating multivariate geometric anisotropy}
	\label{sec:mvga}
	For a population of $P$ LGCPs, we specify the multivariate dependence through the covariance structure of the $P$-dimensional GRF that drives the $P$ conditionally independent intensity processes. We extend the definition of geometric anisotropy in \eqref{eq:geoanisotcov} and we define the following family of geometric anisotropic auto- and cross-covariance functions:
	\begin{equation*}
	C_{pq}(h) = C_{0,pq}\left(\sqrt{h^T\Sigma_{pq}^{-1} h}\right),
	\hspace{25pt}
	p,q = 1,\ldots,P,
	\;\;
	h\in\mathbb{R}^d,
	\end{equation*}
	for some corresponding family of isotropic covariance functions $\{C_{0,pq}(\|h\|);p,q=1,\ldots,P\}$, and for a collection of deformation matrices $\{\Sigma_{pq}; p,q=1,\ldots,P\}$, where $\Sigma_{pq}$ is defined in terms of the parameter pair $(\theta_{pq},\zeta_{pq})$ according to \eqref{eq:geoanisotSigmadef}.

	This framework will allow for the possibility of distinct geometric anisotropies in each of the marginal processes. Such processes can be used to model, for example, bivariate point patterns in which each component displays elliptical clustering at different orientations, or with differing degrees of ellipticity. Care must be taken in specifying the parameters for the cross-covariance functions $C_{pq}$, however, in order to ensure a valid multivariate model. In the spatial domain, we require the matrix of covariance functions $(C_{pq}(h))_{p,q=1}^P$ to be nonnegative definite for all $h\in\mathbb{R}^d$; the equivalent requirement in the spectral domain is that the matrix of spectral densities $(f_{pq}(\omega))_{p,q=1}^P$ is nonnegative definite for all $\omega\in\mathbb{R}^d$. If we consider the bivariate dependence structure for two processes $X_p$ and $X_q$, $p,q=1,\ldots,P$, then we can see that the restriction in the spectral domain is equivalent to requiring the magnitude squared coherence $|\gamma(\omega)|^2$ to be bounded above by 1, where the complex coherence is defined as in \eqref{eq:compcoherence}. This restriction can alternatively, and unsurprisingly, be written at every frequency as 
	\begin{equation}
	\label{eq:crossspectrumUB}
	0\le |f_{pq}(\omega)| \le \left\{f_{pp}(\omega)f_{qq}(\omega)\right\}^{\frac12},\qquad\omega\in\mathbb{R}^d,
	\end{equation}
	and this gives an upper bound on the magnitude of the cross-spectrum. This upper bound is displayed in Figure \ref{fig2:MVGAspectra} for a bivariate process with distinct marginal geometric anisotropies. By considering the behaviour of \eqref{eq:crossspectrumUB} over the two-dimensional Fourier domain, we can now make some general comments about the level of dependence between components in a bivariate geometric anisotropic LGCP; this discussion also applies to pairwise dependences in multivariate LGCPs of higher dimension. For the remainder of this subsection, the only assumption that we make is that each autospectrum and cross-spectrum in the bivariate process is decreasing for increasing frequencies $\omega$. In particular, the following discussion is valid for any family of spectral densities that satisfies this assumption. 

	The inequality in \eqref{eq:crossspectrumUB} implies that, for any two processes, between-process dependence can only be non-negligible at those frequencies that contribute significantly to the marginal dependence in {\em both} processes. 
	For two processes with distinct marginal geometric anisotropies, this restriction impacts the high-frequency behaviour of the bivariate process more than the low-frequency behaviour This can be seen by considering the spectra displayed in Figure \ref{fig2:MVGAspectra}: when constructing the upper bound for the cross-spectrum according to \eqref{eq:crossspectrumUB}, the high-frequency contributions of each of the autospectra are killed by the negligible power at the same frequency in the other autospectrum; the contrasting behaviour of the marginal processes at high frequencies kills any high-frequency dependence \textit{between} the processes. 
	As a result, for any two processes that display contrasting anisotropic behaviour, significant between-process dependence will be more evident at low frequencies, or large spatial scales. 

	Due to our modelling assumption of geometric anisotropy in the cross-dependence structure, the cross-spectrum will have elliptical contours of equal power density. From Figure \ref{fig2:MVGAspectra} we can also see that the elliptical geometries of the autospectra can dictate a nontrivial geometric structure for the upper bound of the cross spectrum. 
	For any given pair of marginal spectra, and thus a given upper bound to the corresponding cross-spectrum, the ellipticity of the true cross-spectrum will therefore impact its permissible coverage of the frequency space, as its elliptical structure must fit within the upper bound's nontrivial geometry. Indeed, we can see from Figure \ref{fig2:MVGAspectra} that, in order for our elliptical cross-spectrum to extend further into the higher-frequency regions of the Fourier space, the ellipticity of the cross-spectrum should be more pronounced; if we were to assume a more isotropic cross-dependence structure, then the non-negligible cross-spectrum would be more restricted to the low-frequency region around the origin. Since the overall power in the cross-process dependence is obtained by integrating the cross-spectrum over the entire Fourier domain, this gives us a link between the power and the degree of anisotropy in the cross-process dependence. We will formalise this relationship towards the end of the next section, in the context of a Mat\'ern specification for our multivariate dependence structure.

	\begin{figure}
		\centering
		\begin{subfigure}{0.275\textwidth}
			\includegraphics[width=\textwidth]{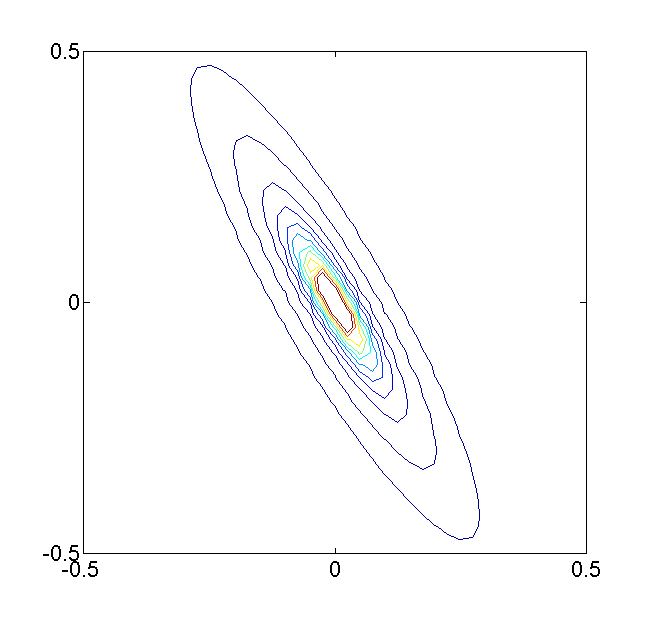}
			\label{subfig:MVGAautospec1}
		\end{subfigure}
		\begin{subfigure}{0.275\textwidth}
			\includegraphics[width=\textwidth]{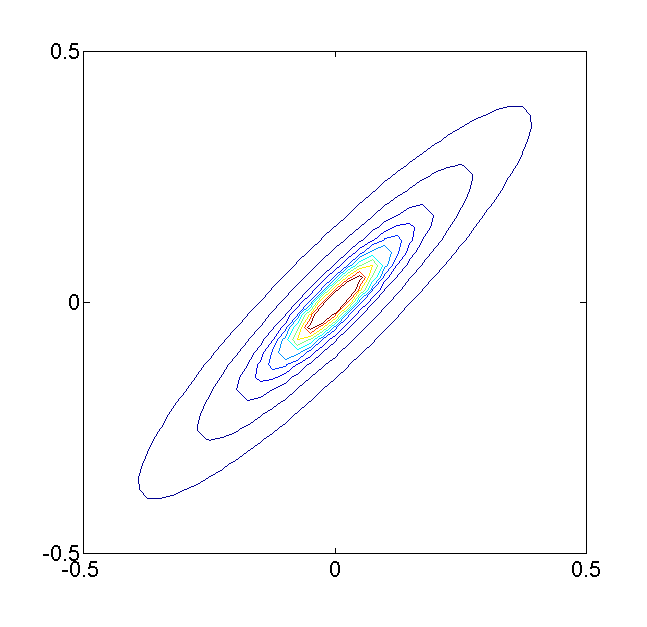}
			\label{subfig:MVGAautospec2}
		\end{subfigure}
		\begin{subfigure}{0.3\textwidth}
			\includegraphics[width=\textwidth]{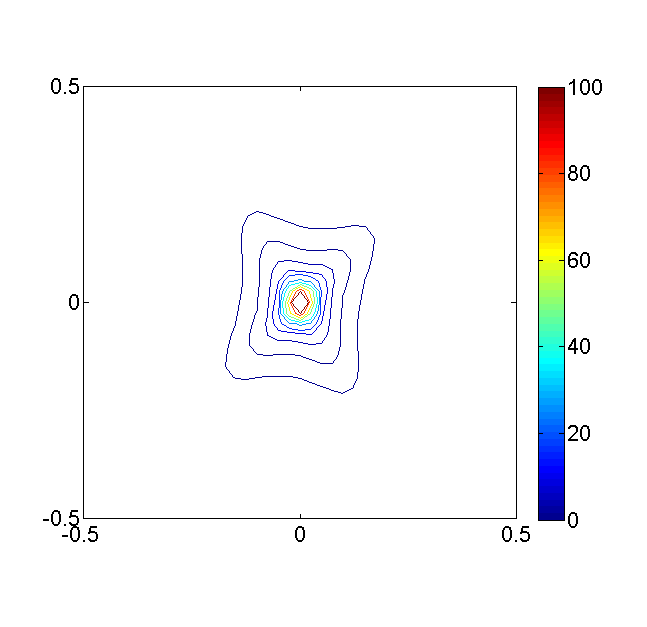}
			\label{subfig:MVGAcrossspecUB}
		\end{subfigure}
		\vspace{-20pt}
		\caption{Geometric anisotropic autospectra (left and centre) for a bivariate GRF, along with the upper bound on the corresponding cross-spectrum (right), as given in equation \eqref{eq:crossspectrumUB}.}
		\label{fig2:MVGAspectra}
	\end{figure}

	\subsection{A multivariate Mat\'ern correlation structure}
	\label{subsec:MVMaterncorr}
	The Mat\'ern family of correlation functions \citep{Stein99,Guttorp06} provides a flexible route to modelling multi-scale dependence in stationary spatial processes. For univariate random fields, one can use a single three-parameter covariance function to replicate dependence structures that act over any positive scale, whilst additionally controlling the smoothness of any realisations. The flexibility of this model has made it the tool of choice for modelling univariate processes in the spatial statistics literature, and there has naturally been a great deal of interest in extending its use to the multivariate setting. \citet{Gneiting10} and \citet{Apanasovich12} have recently addressed this interest, proposing the use of a Mat\'ern function to describe all auto- and cross-covariances for a multivariate isotropic stationary random field. This work has been further extended by \citet{Kleiber12}, who accommodate nonstationarity by allowing the Mat\'ern parameters to vary with respect to location; this allows for the possibility of local anisotropic behaviour, as well as local variances and smoothnesses. We will incorporate elements from these approaches in our modelling framework, though we retain an assumption of stationarity. Our aim is to ascertain whether observed second-order point process characteristics can be modelled independently of first-order covariates; this goal would be best served under assumptions of stationarity in the underlying GRF.

	We develop the stationary Mat\'ern covariance structure for the multivariate GRF $S(x)$, ensuring that all auto- and cross-covariances have a valid Mat\'ern form. The multivariate Mat\'ern model was introduced by \citet{Gneiting10}, who established necessary and sufficient conditions for the validity of the bivariate model, and sufficient conditions for the validity of a restricted subclass of the multivariate ($P\ge3$) model. This work was extended by \citet{Apanasovich12}, who relax the restrictions on the multivariate model and provide sufficient conditions for its validity for any dimension $P\ge1$. We present a class of Mat\'ern auto- and cross-covariance functions that can accommodate multivariate geometric anisotropy, and we adapt the work of \citet{Apanasovich12} in order to provide sufficient conditions for its validity, for $P\ge1$.

	Following \citet{Gneiting10}, we define the isotropic multivariate Mat\'ern covariance function to be
	\begin{eqnarray}
	\label{eq:mvMatern}
	\nonumber
	C_{0,pq}(\|h\|;\alpha_{pq},\nu_{pq},\sigma_{pq}) & = & \frac{\sigma_{pq}}{2^{\nu_{pq}-1}\Gamma\left(\nu_{pq}\right)} \left(\frac{2\displaystyle\sqrt{\nu_{pq}}}{\alpha_{pq}} \|h\|\right)^{\nu_{pq}} \mathcal{K}_{\nu_{pq}}\left(\frac{2\displaystyle\sqrt{\nu_{pq}}}{\alpha_{pq}} \|h\|\right),
	\hspace{10pt}
	h\in\mathbb{R}^d,
	\\
	\end{eqnarray}
	where $\mathcal{K}_{\nu}(\cdot)$ is the modified Bessel function of the second kind \citep[][pp.374--379]{Abramowitz65}. Here, $\sigma_{pq}\in\mathbb{R}$ ($\sigma_{pp}>0$) is the zero-lag covariance between field components $S_p$ and $S_q$, and $\alpha_{pq}>0$ and $\nu_{pq}>0$ are scale and smoothness parameters, respectively. The latter two parameters control the rate of decay of covariance between the same two processes with respect to distance. As a scale parameter, $\alpha_{pq}$ determines the `practical range' of the covariance function, i.e.~the separation distance at which $S_p$ and $S_q$ may be considered approximately independent. The smoothness parameter $\nu_{pq}$ determines the shape of the covariance function, and in particular the speed with which it decays close to the origin. For the marginal processes, $\nu_{pp}$ clearly then controls the smoothness of the realisations; indeed, the marginal process $X_p$ will be $m$ times mean-square differentiable if and only if $\nu_{pp}>m$.

	Throughout the literature, the Mat\'ern covariance function has been defined using a variety of parametric forms, with the three parameters interacting in a different manner in each specification. The predominant material difference between the parameterisations is the formulation of the term used to scale the absolute distance $\|h\|$. The inverse of this distance-scaling factor is also known as the correlation length and this is proportional to the practical range; as can be seen from \eqref{eq:mvMatern}, in the current parameterisation, the correlation length is equal to $\alpha_{pq}/2\sqrt{\nu_{pq}}$. For alternative parameterisations of the model where the correlation length is independent of $\nu_{pq}$, it is often found that the effects of $\alpha_{pq}$ and $\nu_{pq}$ on the practical range and shape of $C_{0,pq}(\|h\|;\alpha_{pq},\nu_{pq},\sigma_{pq})$ cannot be well separated. The parameterisation of the Mat\'ern function given in \eqref{eq:mvMatern}, attributable to \cite{Handcock94} in the univariate scenario, is chosen to allow maximal separation of the roles of $\alpha_{pq}$ and $\nu_{pq}$ in determining the second-order behaviour of $S(x)$ and, ultimately, the resulting point process $X$. 

	As $\alpha_{pq}$ increases for fixed $\nu_{pq}$, so too will the practical range of $C_{0,pq}(\|h\|;\alpha_{pq},\nu_{pq},\sigma_{pq})$. This will increase the maximum distance at which one can expect to find cross-process aggregation and segregation of points in $X_p$ and $X_q$. Note that the corresponding effect for the marginal scale parameters is that an increase (decrease) in $\alpha_{pp}$ will result in an increase (resp.~decrease) in the width of the observed clusters in $X_p$. Recall that the smoothness parameter controls the shape of the covariance function; as $\nu_{pq}$ increases, $C_{0,pq}(\|h\|;\alpha_{pq},\nu_{pq},\sigma_{pq})$ becomes smoother around $\|h\|=0$. As a result of our parameterisation, as $\nu_{pq}$ increases for fixed $\alpha_{pq}$, \eqref{eq:mvMatern} will increase for small values of $\|h\|$ and decrease for large values of $\|h\|$; the distribution of variance shifts from high scales to low scales. Thus, where the scale parameter $\alpha_{pq}$ determines the width of areas in which processes $X_p$ and $X_q$ will have similar intensities, $\nu_{pq}$ will determine {\it how} similar these intensity processes are within these regions. 

	Having established the Mat\'ern form of the auto- and cross-covariances for a multivariate isotropic GRF, we now generalise to allow for anisotropic multivariate covariance structures. Recall from Section \ref{sec:mvga} that we obtain our geometric anisotropic (cross-)covariance function by applying the deformation matrix $\Sigma_{pq}$:
	\begin{eqnarray}
	C_{pq}(h;\alpha_{pq},\nu_{pq},\sigma_{pq},\Sigma_{pq}) & = & \frac{\sigma_{pq}}{2^{\nu_{pq}-1}\Gamma\left(\nu_{pq}\right)} \left(\frac{2\sqrt\nu_{pq}}{\alpha_{pq}}\left\|\Sigma_{pq}^{-1/2}h\right\|\right)^{\nu_{pq}} \mathcal{K}_{\nu_{pq}}\left(\frac{2\sqrt\nu_{pq}}{\alpha_{pq}}\left\|\Sigma_{pq}^{-1/2}h\right\|\right),\nonumber\\
	\label{eq:mvgaMatern}
	\end{eqnarray}
	which is defined for any $h\in\mathbb{R}^d$. 
	
	Recall that we require the matrix $\left(C_{pq}(h;\alpha_{pq},\nu_{pq},\sigma_{pq},\Sigma_{pq})\right)_{p,q=1}^P$ to be nonnegative definite for all $h\in\mathbb{R}^d$, in order for \eqref{eq:mvgaMatern} to define a valid multivariate covariance model.	Satisfaction of this requirement can be guaranteed by placing the following conditions on the cross-covariance parameters $\{\alpha_{pq},\nu_{pq},\sigma_{pq},\theta_{pq},\zeta_{pq}, p\ne q\}$.


	\begin{condition}
	\label{cond:nu}
	There exists a nonnegative constant  $\Delta_\nu$ such that $\nu_{pq} - (\nu_{pp}+\nu_{qq})/2 = \Delta_{\nu}(1-A_{\nu,pq})$, $p,q=1,\ldots,P$, where $A_\nu$ is a valid $P\times P$ correlation matrix, with entries $0\le A_{\nu,pq}\le1$.
	\end{condition}
	\begin{condition}
	\label{cond:alpha}
	The matrix with elements $-4\nu_{pq}/\alpha_{pq}^2$, $p,q=1,\ldots,P$, is {\em conditionally} nonnegative definite. This is a weaker assumption than that of nonnegative definiteness, and it may be satisfied by a matrix containing only negative elements.
	\end{condition}

	\begin{condition}
	\label{cond:sigma}
	The matrix with elements 
	$$
	\frac{|\Sigma_{pq}|^{1/2}\sigma_{pq}\Gamma(\nu_{pq}+d/2)}{\pi^{d/2}\Gamma(\frac{\nu_{pp}+\nu_{qq}}{2}+\frac{d}{2})\Gamma(\nu_{pq})}\left(\frac{4\nu_{pq}}{\alpha_{pq}^2}\right)^{\Delta_{\nu}+\frac{\nu_{pp}+\nu_{qq}}{2}},
	\qquad
	p,q=1,\ldots,P,
	$$ 
	is nonnegative definite.
	\end{condition}

	\begin{condition}
	\label{cond:Sigma}
	The matrix with elements $-\|\Sigma_{pq}^{1/2}\omega\|^2$, $p,q=1,\ldots,P$, is conditionally nonnegative definite for any $\omega\in\mathbb{R}^d$. 
	\end{condition}

	\begin{proposition}
	\label{prop:mvgaMatern}
	For $p,q=1,\ldots,P$, let $\alpha_{pq}>0$, $\nu_{pq}>0$, $\sigma_{pq}\in\mathbb{R}$, $\theta\in[0,2\pi)$ and $\zeta\in(0,1]$. Then the multivariate geometric anisotropic Mat\'ern function \eqref{eq:mvgaMatern} specifies a valid multivariate covariance model if Conditions \ref{cond:nu}-\ref{cond:Sigma} are met. 
	\end{proposition}
	The proof of Proposition \ref{prop:mvgaMatern} is given in the Appendix, 
	and follows a similar argument to the proof of Theorem 1 of \citet{Apanasovich12}.

	\begin{remark}
	\label{rmk:zeta}
	If Condition \ref{cond:Sigma} holds, then the $P\times P$ matrix with $(p,q)$-element $\left|\Sigma_{pq}\right|^{-1/2} = \zeta_{pq}^{-1},$ will be nonnegative definite; in particular, we can deduce $\zeta_{pq}^2\ge\zeta_{pp}\zeta_{qq}$, for all $p,q=1,\ldots,P.$
	\end{remark}

	Conditions \ref{cond:nu}-\ref{cond:Sigma} are similar in spirit to those placed by \citet{Apanasovich12} on the Mat\'ern parameters in order to guarantee a valid multivariate dependence structure in an isotropic setting. In the simpler isotropic framework, the three conditions specified by \citet{Apanasovich12} are sufficient to guarantee nonnegative definiteness of the resulting spectral density, and also to guarantee that all absolute zero-lag cross-correlations are bounded above by one. In the more general geometric anisotropic setting, we require a more extensive specification. Conditions \ref{cond:nu}-\ref{cond:Sigma}, above, are sufficient to guarantee nonnegative definiteness of the geometric anisotropic spectral density, and are also sufficient for the absolute colocated cross-correlations to be bounded above by 1. 

	These conditions constitute a set of implicit relationships that, between them, specify a valid multivariate geometric anisotropic LGCP. We will now provide explicit restrictions on the cross-dependence parameters in terms of the marginal dependence parameters. This will allow users to sequentially construct a valid multivariate model by first specifying the marginal covariances, and then conditionally specifying the cross-covariance structures. This sequential approach to model construction will also be reflected in our model-fitting procedures in Section \ref{sec:modelfitting}.

	Trivial rearrangement of Condition \ref{cond:nu} yields an explicit expression for $\nu_{pq}$ in terms of the corresponding marginal values. In Remarks \ref{rmk:alpha}-\ref{rmk:equicorr}, below, we provide similar constructions for the cross-covariance parameters $\alpha_{pq},$ $\sigma_{pq}$, $\theta_{pq}$ and $\zeta_{pq}$, such that Conditions \ref{cond:nu}-\ref{cond:Sigma} may be satisfied. The proofs for Remarks \ref{rmk:alpha}-\ref{rmk:equicorr} are given in the Appendix. 

	\begin{remark}
	\label{rmk:alpha}
	Condition \ref{cond:alpha} is satisfied by the parameters $\{\nu_{pq},\alpha_{pq}\;;\;p,q=1,\ldots,P\}$ if 
	\begin{equation*}
	\frac{4\nu_{pq}}{\alpha_{pq}^2} = \frac12\left(\frac{4\nu_{pp}}{\alpha_{pp}^2}+\frac{4\nu_{qq}}{\alpha_{qq}^2}\right) + \Delta_{\alpha}\left(1-A_{\alpha,pq}\right),
	\end{equation*}
	for some constant $\Delta_{\alpha}\ge0$ and for some $0\le A_{\alpha,pq}\le1$ that form a valid correlation matrix. This remark is also made by \citet{Apanasovich12} in their chosen Mat\'ern parameterisation.
	\end{remark}

	\begin{remark}
	\label{rmk:sigma}
	Condition \ref{cond:sigma} is satisfied by the parameters $\{\nu_{pq},\alpha_{pq},\zeta_{pq},\sigma_{pq}\;;\;p,q=1,\ldots,P\}$ if
	\begin{equation*}
	\sigma_{pq} = \frac{\pi^{d/2} V_{p}V_{q}A_{\sigma,pq} }{\zeta_{pq}} 	
	\left(\frac{4\nu_{pq}}{\alpha_{pq}^2}\right)^{-\Delta_{\nu}-\frac{\nu_{pp}+\nu_{qq}}{2}}
	\frac{\Gamma(\frac{\nu_{pp}+\nu_{qq}}{2}+\frac{d}{2})\Gamma(\nu_{pq})}{\Gamma(\nu_{pq}+\frac{d}{2})}
	\qquad p,q=1,\ldots,P,
	\end{equation*}
	for constants $V_p,V_q\ge0$ and for some $A_{\sigma,pq}\in[-1,1]$ that form a valid correlation matrix.
	\end{remark}

	\begin{remark}
	\label{rmk:Sigma}
	Condition \ref{cond:Sigma} is satisfied by the deformation matrices $\{\Sigma_{pq}\;;\;p,q=1,\ldots,P\}$ if their diagonal elements $\left[\Sigma_{pq}\right]_{ii}$, can be written 
	$$
	\left[\Sigma_{pq}\right]_{ii} = \frac12\left[\Sigma_{pp}+\Sigma_{qq}\right]_{ii} + \Delta^{(i)}_{\Sigma}\left(1-A^{(i)}_{\Sigma,pq}\right),\qquad i=1,2.
	$$
	\end{remark}

	\begin{remark}
	\label{rmk:equicorr}
	For small $P$, we can follow the lead of \citet{Apanasovich12} and use equicorrelated matrices $A^{(i)}_{\Sigma}$, $i=1,2$, setting $A^{(i)}_{\Sigma,pq}=\rho^{(i)}_{\Sigma}$, $p\ne q$; in this scenario, for the sake of identifiability, we redefine $\Delta_{\Sigma}^{(i)}:=\Delta_{\Sigma}^{(i)}(1-\rho_{\Sigma}^{(i)})$, $i=1,2$. 
	\end{remark}

	Conditions \ref{cond:nu}-\ref{cond:Sigma}, along with Remarks \ref{rmk:alpha}-\ref{rmk:Sigma}, indicate a sequential approach to specifying a valid multivariate geometric anisotropic Mat\'ern covariance structure in practice. As mentioned previously, Condition \ref{cond:nu} and Remarks \ref{rmk:alpha}-\ref{rmk:Sigma} suggest that one must specify the parameters for the marginal covariance function before conditionally specifying the parameters for each cross-covariance function. These statements also indicate that, within each individual component of the joint model, i.e.~for fixed $p,q$, there is a particular order in which the five parameters $(\theta_{pq},\zeta_{pq},\alpha_{pq},\nu_{pq},\sigma_{pq})$ should necessarily be specified. From Remark \ref{rmk:sigma}, we can see that, for each $(p,q)$ pairing, the specification of the Mat\'ern power parameter $\sigma_{pq}$ is dependent upon the corresponding ratio of anisotropy $\zeta_{pq}$, as well as the other Mat\'ern parameters, $\alpha_{pq}$ and $\nu_{pq}$, and Remark \ref{rmk:Sigma} indicates that the anisotropy parameters $(\theta_{pq},\zeta_{pq})$ should be jointly specified. In addition, Condition \ref{cond:nu} and Remark \ref{rmk:alpha} indicate that the smoothness parameter $\nu_{pq}$ should be specified before the scale parameter $\alpha_{pq}$. We conclude that, for each marginal or bivariate component of the joint covariance model, the anisotropy parameters should be specified before the Mat\'ern parameters, with the Mat\'ern smoothness, scale and power parameters being specified third, fourth and fifth, respectively.

	We conclude this section by considering the limitations placed on the zero-lag cross-correlation coefficients $\rho_{pq}:=\sigma_{pq}/\sqrt{\sigma_{pp}\sigma_{qq}}$. By rearranging Condition \ref{cond:sigma}, we can write:
	\begin{equation}
	\label{eq:colocxcorrs}
	\rho_{pq}^{2} = \frac{\sigma_{pq}^2}{\sigma_{pp}\sigma_{qq}} \le \prod_{i=1}^4\tau_{pq}^{(i)}\le1,
	\end{equation}
	with 
	\begin{equation*}
	\tau_{pq}^{(1)} = \frac{\mathcal{B}^2(\nu_{pq},\frac{d}{2})}{\mathcal{B}^2(\frac{\nu_{pp}+\nu_{qq}}{2},\frac{d}{2})},
	\hspace{15pt}
	\tau_{pq}^{(2)} = \left[\frac{\frac{4\nu_{pp}}{\alpha_{pp}^2}\frac{4\nu_{qq}}{\alpha_{qq}^2}}{\left(\frac{4\nu_{pq}}{\alpha_{pq}^2}\right)^2}\right]^{\Delta_{\nu}},
	\end{equation*}
	\begin{equation*}
	\tau_{pq}^{(3)} = \frac{\Gamma^2(\frac{\nu_{pp}+\nu_{qq}}{2})\left(\frac{\alpha_{pq}^2}{4\nu_{pq}}\right)^{\nu_{pp}+\nu_{qq}}}
	{\Gamma(\nu_{pp})\left(\frac{\alpha_{pp}^2}{4\nu_{pp}}\right)^{\nu_{pp}}\Gamma(\nu_{qq})\left(\frac{\alpha_{qq}^2}{4\nu_{qq}}\right)^{\nu_{qq}}},
	\hspace{10pt}
	\tau_{pq}^{(4)} = \frac{|\Sigma_{pp}|^{1/2}|\Sigma_{qq}|^{1/2}}{|\Sigma_{pq}|} = \frac{\zeta_{pp}\zeta_{qq}}{\zeta_{pq}^2},
	\end{equation*}
	where $\mathcal{B}(\cdot,\cdot)$ is the Beta function \citep{Abramowitz65}. The first inequality in \eqref{eq:colocxcorrs} is directly implied by Condition \ref{cond:sigma}. The second inequality in \eqref{eq:colocxcorrs} can be shown componentwise: By Remark \ref{rmk:zeta}, Condition \ref{cond:Sigma} ensures that $\tau_{pq}^{(4)}\le1$, and as noted by \citet{Apanasovich12} in the isotropic framework, Conditions \ref{cond:nu} and \ref{cond:alpha} are sufficient to guarantee that $\tau_{pq}^{(i)}\le 1$, $i=1,2,3$.

	In the isotropic framework, $\tau_{pq}^{(4)}=1$, and we are left with the limitations noted by \citet{Apanasovich12}: the zero-lag cross-correlation will be bounded above by 1 when the corresponding univariate isotropic processes share identical Mat\'ern parameters. When the marginal parameter specifications differ, this upper bound will decrease as the smoothness and inverse correlation length of the cross-covariance structure depart from the arithmetic mean of the corresponding marginal quantities.

	In our more general anisotropic framework, we can see from $\tau_{pq}^{(4)}$ in \eqref{eq:colocxcorrs} that the upper bound on the colocated cross-correlations will also be affected by the relationship between the cross-covariance ratio of anisotropy $\zeta_{pq}$ and the ratios of anisotropy in the corresponding marginal covariance structures. 
	If we assume Condition \ref{cond:Sigma} to hold, then by Remark \ref{rmk:zeta}, $\zeta_{pq}$ will be restricted to the closed interval $[\zeta_{pp}^{1/2}\zeta_{qq}^{1/2},1]$. If $\zeta_{pq}=\zeta_{pp}^{1/2}\zeta_{qq}^{1/2}$, then $\tau_{pq}^{(4)}$ will reduce to 1, and the upper bound of the colocated cross-correlation $\rho_{pq}$ will behave as in the isotropic framework, i.e.~as described above. Increasing $\zeta_{pq}$ away from this geometric mean, however, will decrease $\tau_{pq}^{(4)}$, which will in turn shrink the upper bound on $\rho^2_{pq}$, given in \eqref{eq:colocxcorrs}.
	In other words, as the ellipticity of the cross-covariance function becomes less pronounced, the maximum possible degree of zero-lag correlation between the two components of the field will decrease. This formalizes the relationship between the power and the anisotropy of the cross-process dependence, discussed at the end of Section \ref{sec:mvga}.

%


	\section{Fitting the Model}
	\label{sec:modelfitting}

	\subsection{Parameter Estimation Procedure}
	\label{subsec:paramest}

	In order to fit our parametric model to an observed multitype point pattern, we must estimate both the marginal and joint anisotropy parameters $\{\theta_{pq},\zeta_{pq};p,q=1,\ldots,P\}$, as well as the parameters that specify the mean and Mat\'ern covariance structure of the underlying Gaussian random field, $\{\mu_p,\alpha_{pq},\nu_{pq},\sigma_{pq};p,q=1,\ldots,P\}$. At a high level, we follow the approach of \citet{Moller14}, who fit a univariate version of our model by first estimating the angle and ratio of anisotropy in the observed data, before using these estimates to back-transform the data into an isotropic framework. The resulting `isotropised' point pattern is then used to estimate the mean parameters and the Mat\'ern parameters. Our approach to each component of this two-stage model-fitting procedure will differ from the methods of \citet{Moller14}, however. We use an approach to estimating anisotropy that is less sensitive to user-specified tuning parameters, which we adapt from the work of \citet{Rajala16}, and we use a more automatable approach to estimating the mean and Mat\'ern parameters, which we develop from the work of \citet{Tanaka08}.

	In developing our parameter estimation methodology, we are faced with the question of whether to put measures into place to guarantee that the fitted model satisfies Conditions \ref{cond:nu}-\ref{cond:Sigma}, therefore ensuring validity of the multivariate dependence structure. This is the approach taken by \citet{Apanasovich12} for fitting multivariate isotropic Mat\'ern GRFs; they fit the marginal dependence structures then use the estimated marginal parameters to restrict the parameter subspace for the Mat\'ern cross-covariances. Since Conditions \ref{cond:nu}-\ref{cond:Sigma} are sufficient, and not necessary, the resulting restriction on the joint dependence structure could be overstated, potentially resulting in inconsistent estimators for the Mat\'ern cross-covariance parameters. Under the assumption that the smoothness is known, however, the power and scale parameters for a univariate Mat\'ern covariance function cannot be consistently estimated under infill asymptotics \citep{Zhang04}; consistency can only be achieved by increasing the observation window $W$. As noted by \citet{Apanasovich12}, constraining $\sigma_{pq}^2$ and $\alpha_{pq}^2$ ($p\ne q$) conditional on their corresponding marginal values therefore provides no additional penalty in terms of estimator consistency when assuming a fixed observation window. Furthermore, the numerical tests of \citet{Apanasovich12} show that reasonable accuracy can indeed be obtained when using this constrained approach to parameter estimation; this approach therefore warrants examination in the current framework. In order to avoid compromising the consistency of the anisotropy estimators, we do not use our conditions from Section \ref{sec:model} to restrict the parameter pair $(\theta_{pq},\zeta_{pq})$, $p\ne q$.

	\subsection{Estimating the Anisotropy Parameters}

We focus first on quantifying the anisotropy present in both the marginal and joint dependence structures in a multi-type point pattern. \citet{Moller14} estimate the angle of anisotropy in a univariate geometric anisotropic point pattern by finding the angle $\phi$ at which the $r$-integrated difference between the anisotropic pair correlation function $g^a(r,\phi)$ and its phase-shifted self $g^a(r,\phi+\pi/2)$, is maximised. This is achieved by estimating $g^a(r,\phi)$ over a discrete lattice of polar coordinates $(r,\phi)$, and numerically approximating the required integral in $r$. Accuracy of the resulting estimator is therefore sensitive to the resolution of the polar lattice, as well as the choice of two bandwidth parameters used in estimating the anisotropic pair correlation function; for details of these bandwidth parameters, see \citet{Moller14}. Finally, use of this estimation method is also dependent on the assumption that the isotropic pair correlation function is strictly decreasing. Whilst this assumption holds true for our assumed Mat\'ern model, it can be violated by real data. The approach we detail below is more widely applicable, as it does not depend on such an assumption, and it is also less sensitive to subjective choices of bandwidth parameters.

We adopt and adapt the method introduced by \citet{Rajala16} for estimating the angle of anisotropy: we adopt this method for characterising anisotropy in the marginal covariance structures, and we adapt it for estimating the angle of anisotropy in the cross-covariance structures. For the sake of generality, we describe the procedure for estimating $\theta_{pq}$, $p\ne q$. We start by constructing the point pattern formed by the difference vectors $\{
x_{p,i}-x_{q,j}\;;\; i=1,\ldots,n_p, j=1,\ldots,n_q\}$; this is the (bivariate) Fry process \citep{Fry79}, and when $p=q$, this will be rotationally symmetric of order 2, about the origin.
The Fry process is useful here as its first-order properties will reflect the second-order properties of the original point pattern. We can therefore estimate any second-order anisotropy in the original bivariate point pattern by estimating the anisotropy in the intensity of the bivariate Fry process. 


Dividing the polar plane into a selected number, $n_F$, of distinct sectors, and for $l\in L\subset\mathbb{N}$, we collect the $l$th nearest Fry point in each sector into a set, $G_l$, of $n_F$ points, such that each $G_l$ sketches out a noisy contour around the origin, and such that the intensity of the Fry process is reflected in the proximity of the $G_l$s to one another. For point patterns that display segregation, the anisotropy in the joint second-order dependence structure will be shared by the contours of the intensity field for the Fry process; for aggregated point patterns, the angle of anisotropy will be phase-shifted by $\pi/2$ in the Fry process. In order to quantify the anisotropy in the original point pattern then, we can treat the $G_l$s as sampled versions of the Fry intensity's contours, and assuming Gaussian measurement error we can infer the corresponding true contours using adjusted ordinary least squares, and subsequently derive the angle of anisotropy in the original point pattern. 
For full technical details of this method, we direct the reader to \citet{Rajala16}. 

For each marginal process, as described by \citet{Moller14}, we can transform the observed point pattern $X_p\cap W$ and the corresponding observation window $W$ by assuming fixed values for $\theta\in[0,2\pi)$ and $\zeta\in[0,1]$:
\begin{eqnarray*}
X_{p,\theta,\zeta} = X_p R_{\theta}^T 
\left(
\begin{array}{cc}
1	&	0	\\
0	&	\zeta^{-1}
\end{array}
\right)
,\qquad&&\qquad 
W_{\theta,\zeta} = W R_{\theta}^T 
\left(
\begin{array}{cc}
1	&	0	\\
0	&	\zeta^{-1}
\end{array}
\right)
.
\end{eqnarray*}

If the chosen values of $\theta$ and $\zeta$ are equal to the values that describe the anisotropy of $X_p$, then the transformed point process $X_{p,\theta,\zeta}$  will be isotropic and the corresponding anisotropic pair correlation function $g^a_{pp,\theta,\zeta}(r,\phi)$ will be constant with respect to its second argument. This motivates our chosen method for estimating the \textit{marginal} anisotropy ratios $\zeta_{pp}$, which we also adopt from the work of \citet{Rajala16}. 

Following \citet{Rajala16}, we define the following directional discrepancy statistic:
\begin{equation}
\label{eq:intdiffK}
V_{pp,\theta}(\zeta) = \int_{b_1}^{b_2} \left[K^a_{pp,\theta,\zeta}(r,0) - K^a_{pp,\theta,\zeta}(r,\pi/2)\right] dr,
\end{equation}
where
\begin{equation}
\label{eq:sectorKfun}
K^a_{pp,\theta,\zeta}(r,\phi) = \int_0^r g^a_{pp,\theta,\zeta}(s,\phi)ds
\end{equation}
is the sector-$K$-function, an anisotropic variant of Ripley's $K$-function, evaluated on the isotropised point pattern $X_{p,\theta,\zeta}$. To estimate the marginal ratio of anisotropy $\zeta_{pp}$, we back-transform our observed point pattern using the estimated angle of anisotropy $\hat{\theta}_{pp}$ and a sequence of candidate ratios $\{\zeta_{pp,k}:=k\zeta_{max}/(1+n_{\zeta}),\; k=1,\ldots,n_{\zeta}\}$, for some user-defined upper bound $\zeta_{max}$. We then choose $\hat{\zeta}_{pp} = \zeta_{pp,k}\in(0,\zeta_{max})$ to be the candidate value that minimises the estimate $\hat{V}_{pp,\hat{\theta}_{pp}}(\zeta_{pp,k})$. Note that, although we defined $\zeta_{pp}\in(0,1)$ in Section \ref{subsec:GALGCPs}, the sampling variance of the estimated sector-$K$-function can result in an estimated ratio $\hat{\zeta}_{pp}>1$. 

\citet{Moller14} use a similar approach, in effect minimising the directional discrepancy statistic \eqref{eq:intdiffK}, but using the anisotropic pair correlation function in place of the sector-$K$-function. Indeed, it is possible to use any directional second-order statistic in the integrand of \eqref{eq:intdiffK}. We choose to use $K^a_{pp,\theta,\zeta}$ for two reasons. Firstly, the analysis of \citet{Redenback09} suggests that the sector-$K$-function is better-suited to characterising anisotropy than nearest-neighbour statistics; the authors conclude that, for detecting anisotropy in point patterns, statistical tests based on the sector-$K$-function have greater power, in general, than those based on nearest-neighbour orientation statistics. Secondly, estimation of the sector-$K$-function requires the choice of only one tuning parameter, an angular bandwidth, whereas the use of the anisotropic pair correlation function would require the specification of both angular and radial bandwidths.

Our chosen approach to estimating $\zeta_{pp}$ can be extended to the multivariate scenario, where we are interested in the geometric anisotropic cross-dependence exhibited by a given pair of Cox processes $X_p$ and $X_q$. By manipulating the space $\mathbb{R}^d$ on which both processes live, we also manipulate the cross-covariance function $C_{pq}(h)$ that specifies the dependence between the random fields that drive $X_p$ and $X_q$. For each pair of processes, we once again define a discrete set of candidate multivariate anisotropy ratios $\{\zeta_{pq,k}\in(0,\zeta_{max}), k=1,\ldots,n_{\zeta}\}$, and we choose $\hat{\zeta}_{pq} = \zeta_{pq,k}$ for which the estimated value of $V_{pq,\hat{\theta}_{pq}(\zeta_{pq,k})}$ is minimised, where $V_{pq,\hat{\theta}_{pq}(\zeta_{pq,k})}$ is defined through transforming both $X_p\cap W$ and $X_q\cap W$, along with their common observation window $W$.

The above approach to estimating the anisotropy parameters requires the selection of a number of control parameters: the number of sectors $n_F$, into which we partition the Fry process; the number $n_\zeta$ of candidate ratios of anisotropy, as well as their upper bound $\zeta_{max}$; and the limits of integration, $b_1$ and $b_2$ in \eqref{eq:intdiffK}, which we use to calculate $\hat{V}_{pq,\hat{\theta}_{pq}}(\zeta_{pq,k})$ when estimating $\zeta$. As a rule of thumb, and for reasons outlined below, \citet{Rajala16} suggest choosing $n_F\approx\lambda|W|/6$, where 
$\lambda|W|$ is the expected number of points in the original point process. 
We adopt this guideline for choosing $n_F$ when estimating the anisotropy in the marginal processes, and we derive a similar rule of thumb for $n_F$ when estimating the between-process anisotropy, by following the same arguments as \citet{Rajala16}. For the bivariate Poisson process with intensity vector $(\lambda_p,\lambda_q)$ in a circular spatial window $W$, the expected number of bivariate Fry points per sector is approximately $\lambda_p\lambda_q|W|^2/3n_F$. 
Each point in the bivariate process can be expected to contribute if there are at least $(\lambda_p+\lambda_q)|W|$ points per sector, and so we have a bivariate direction count rule of $n_F\approx \lambda_p\lambda_q|W|/3(\lambda_p+\lambda_q)$. 
Selection of both $n_{\zeta}$ and $\zeta_{max}$ is straightforward: $\zeta_{max}$ should be chosen such that $(0,\zeta_{max})$ covers the majority of the sampling distribution of $\zeta_{pq}$, and selection of $n_{\zeta}$ involves a trade-off between accuracy in the resulting estimates and computational expense of the estimation procedure. In Section \ref{sec:implementation}, where we implement our model fitting procedure for both simulated data and tropical rainforest data, we use $\zeta_{max}=2$ and $n_{\zeta}=199$ for estimating all marginal and joint ratios of anisotropy. Choice of the limits of integration, $b_1$ and $b_2$ in \eqref{eq:intdiffK}, is a more subjective task, and should be determined by the range of scales over which dependence (either within, or between processes) is sought to be characterised; these need not be the same for all marginal and cross-dependence relationships being estimated. In Section \ref{sec:implementation}, we detail our choices of these limits of integration.

\subsection{Estimating the Mat\'ern Parameters}
\label{subsec:Maternestimation}
Once we have estimated our anisotropy parameters, we can isotropise the point pattern and its observation window, and use this transformed data to estimate the remaining parameters. In order to ensure that the Mat\'ern parameters satisfy Conditions \ref{cond:nu}-\ref{cond:sigma}, we define $\nu_{pq}$, $\alpha_{pq}$ and $\sigma_{pq}$ according to the specifications in Condition \ref{cond:nu}, Remark \ref{rmk:alpha} and Remark \ref{rmk:sigma}, respectively. 
Techniques for modelling the correlation matrices $A_{\nu}$, $A_{\alpha}$ and $A_{\sigma}$ are discussed by \citet{Apanasovich10}, and the reader is directed there for further details.
When $P$ is small, however, we can simplify our task by assuming the off-diagonal elements of $A_{\alpha},A_{\nu},A_{\sigma}$ to be constant \citep{Apanasovich12}.

In order to estimate both the mean and Mat\'ern parameters, we maximise the Palm log-likelihood, first proposed by \citet{Tanaka08}. 
For estimating the marginal parameters, we use the version of the Palm log-likelihood given by \citet{Dvorak12}, where the inner region correction is proposed to deal with edge effects:
\begin{multline}
\label{eq:palmloglik}
\ell(\lambda_p,\alpha_{pp},\nu_{pp},\sigma_{pp}) \approx \sum_{\substack{x_{p,i}\in X_{p,\theta,\zeta}\cap W_{\theta,\zeta}\setminus R \\x_{p,j}\in X_{p,\theta,\zeta}\cap W_{\theta,\zeta}\\r_{ij}<R}}^{\ne}
\log\left\{\lambda_p g_{pp}(r_{ij};\alpha_{pp},\nu_{pp},\sigma_{pp})\right\}
\\
- \lambda_p|X_p\cap W\setminus R| K_{p}(R;\alpha_{pp},\nu_{pp},\sigma_{pp}),
\end{multline}
where 
$r_{ij}=\|x_{p,i}-x_{p,j}\|$, $K_{p}(r;\alpha_{pp},\nu_{pp},\sigma_{pp})$ is Ripley's univariate $K$-function, which we approximate by numerically integrating the corresponding pair correlation function, and $|X_{p,\theta,\zeta}\cap W\setminus R|$ denotes the number of points in the isotropised pattern $X_{p,\theta,\zeta}$ that lie further than a distance $R$ from the boundary of $W_{\theta,\zeta}$. $R$ is a user-defined tuning parameter that can be objectively set based on the data; this is discussed further in Section \ref{sec:implementation}. As is common in the point pattern literature, we use $\neq$ in the summation notation to indicate summation over pairs of distinct points.

The Palm log-likelihood \eqref{eq:palmloglik} can be analytically maximised with respect to $\lambda_p$, yielding the maximum Palm-likelihood estimate (MPLE) $\hat{\lambda}_p$, and we obtain MPLEs for the remaining marginal Mat\'ern parameters by numerically maximising $\ell(\hat{\lambda}_p,\alpha_{pp},\nu_{pp},\sigma_{pp})$. The MPLE for $\mu_p$ can be subsequently calculated according to \eqref{eq:lamp_mup}.

 We further develop the Palm log-likelihood approach, in order to estimate the parameters for the cross-covariance structure; our bivariate Palm log-likelihood follows a similar construction to the marginal version. First, we obtain the symmetric bivariate Fry process for components $X_p$ and $X_q$, using the inner region correction to deal with edge effects. We then treat this Fry process as an inhomogeneous Poisson process, with intensity equal to a bivariate version of the Palm intensity \citep{Daley08,Prokesova13}, which we define heuristically as follows: for $x$ at distance $r$ from the origin $o$, the occurrence rate of process $q$ at $x\in\left\{\mathbb{R}^2:\|x\|=r\right\}$, assuming there to be a point of process $p$ at the origin, is
 $$
\lambda_{0,pq}(x)dx = \mathbb{P}\left(|X_q \cap dx|=1\big||X_p\cap \{o\}|=1\right),
 $$ 
 where $dx$ is the Lebesgue measure for the infinitesimal set at $x$. Following this definition, we can relate the bivariate Palm intensity to the (isotropic) cross-pair correlation function for the original process: 
 $$
 \lambda_{0,pq}(r) = \lambda_q g_{0,pq}(r),
 $$
 and this allows us to obtain the following bivariate Palm log-likelihood, which can be maximised to obtain estimates for $\alpha_{pq}$, $\nu_{pq}$, $\sigma_{pq}$, $p\ne q$:

\begin{multline}
\label{eq:bivpalmloglik}
\ell(\lambda_p,\lambda_q,\alpha_{pq},\nu_{pq},\sigma_{pq}) \approx \sum_{\substack{x_{p,i}\in X_p\cap W_{\theta,\zeta}\\x_{q,j}\in X_q\cap W_{\theta,\zeta}\\r_{ij}<R}}^{\ne}
\log\left\{(\lambda_p+\lambda_q) g_{pq}(r_{ij};\alpha_{pq},\nu_{pq},\sigma_{pq})\right\}
\\
- \Big(|X_q\cap W\setminus R|\lambda_p+|X_p\cap W\setminus R|\lambda_q\Big) K_{pq}(R;\alpha_{pq},\nu_{pq},\sigma_{pq}),
\end{multline}
where $r_{ij}=\|x_{p,i}-x_{q,j}\|$, $K_{pq}(r;\alpha_{pq},\nu_{pq},\sigma_{pq})$ is Ripley's bivariate $K$-function, and $|X_p\cap W\setminus R|$ denotes the number of observed points in process $p$ that lie further than a distance $R$ from the boundary of the window $R$. By substituting our previous estimates of $\lambda_p$ and $\lambda_q$ into \eqref{eq:bivpalmloglik}, we obtain an expression in terms of the Mat\'ern cross-covariance parameters only. We numerically maximise this expression in $(\alpha_{pq},\nu_{pq},\sigma_{pq})$ over the constrained parameter space described by Condition \ref{cond:nu}, Remark \ref{rmk:alpha} and Remark \ref{rmk:sigma}, and dependent on the corresponding estimated marginal Mat\'ern parameters. As described in Section \ref{subsec:paramest}, the use of constrained optimisation should not affect the consistency of the cross-covariance parameter estimators, however they may display some bias due to the truncation of their supports. 




\section{Implementation}
\label{sec:implementation}

\subsection{Proof of concept simulations}
\label{subsec:proofofconcept}
We demonstrate the validity of the model fitting procedure described in Section \ref{sec:modelfitting}, through a series of Monte Carlo simulation studies. Using the restrictions in Section \ref{sec:model}, we define four distinct bivariate geometric anisotropic LGCPs with valid Mat\'ern covariance structures; the parameter values for each model are given in Table \ref{table:MCsims_fullMatern}. For all four models, the parameters are chosen such that the expected log-intensity for each process component, $\log(\lambda_p)= 6.75$ ($p=1,2$), specifying point patterns with a similar intensity to the ecological data to be considered in Section \ref{subsec:BCIimplementation}. For each of the four fully-specified models, we simulate 500 distinct point patterns on the unit square, $W=[0,1]^2$.

For each model, and for $p,q=1,2$, we executed our parameter estimation procedure as described in Section \ref{subsec:paramest}. 
For both the marginal and cross-dependence relationships, we estimate $\theta_{pq}$ using Fry processes consisting of only those point pairs separated by $r\in(0,0.25)$. Similarly, when estimating $\zeta_{pq}$, we numerically approximate the integral $V_{pq,\hat{\theta}_{pq}}(\zeta)$ as defined in \eqref{eq:intdiffK}, using the limits of integration $b_1=0$, $b_2=0.25$. Approximation of $V_{pq,\hat{\theta}_{pq}}(\zeta)$ involves estimating the sector-$K$-function over a discrete, high-resolution set of distances $r$, using an angular bandwidth parameter which we choose to be $h_{\phi} = \pi/8$ following \citet[][\S 4.3.1]{Rajala18a}, and details of the chosen sector-$K$-function estimator are given in the Appendix. 
In estimating the anisotropy parameters, our choice of interval for $r$ is deliberately large relative to the true scale of dependence in all of our models, as we intend to show that reasonable results can be obtained without prior knowledge of the true scale of dependence in the data.

When estimating the Mat\'ern parameters, despite using a favourable form of the Mat\'ern parameterisation as discussed in Section \ref{subsec:MVMaterncorr}, 
there proved to be insufficient separation of the effects of $\nu_{pq}$ and $\alpha_{pq}$ in practice for both parameters to be allowed to vary freely during estimation. In order to avoid this issue, a common strategy \citep[e.g.][]{Diggle13b} is to restrict $\hat{\nu}_{pq}$ to three candidate values, representing three sufficiently distinct levels of smoothness in the resulting random fields: we seek $\hat{\nu}_{pq}\in\{0.05,0.5,5.0\}$, $p,q=1,2$. For the case $p\ne q$, this candidate vector was further restricted, to ensure that $\hat{\nu}_{12}$ satisfied Condition \ref{cond:nu}. The remaining Mat\'ern parameters were allowed to vary on continuous bounded intervals: $\hat{\alpha}_{pq}\in(0,\alpha^{UB}_{pq})$ and $\hat{\sigma}_{pq}\in(0,\sigma^{UB}_{pq})$. In the marginal cases, $\alpha^{UB}_{pp}=10$ and $\sigma^{UB}_{pp}=50$, $p=1,2$, were chosen such that these constituted generous intervals around the corresponding true values. For estimating the cross-covariance parameters, $\alpha^{UB}_{12}$ and $\sigma^{UB}_{12}$ were chosen to ensure compliance with Conditions \ref{cond:alpha} and \ref{cond:sigma}.

Our implementation was carried out in Matlab, where we used the default interior-point algorithm to carry out constrained maximisation of the Palm-log likelihood with respect to $(\alpha_{pq},\sigma_{pq})$, for each candidate value of $\nu_{pq}$. Since this algorithm requires the user to initialise the parameters being sought, we did so using a computationally inexpensive version of the widely-used minimum contrast method, minimising the difference between the estimated (isotropic) pair correlation function and its closed-form expression across a coarse grid of parameter pairs $(\alpha_{pq},\sigma_{pq})$. 


We also detail our choice of the MPLE tuning parameter $R$. Following the guidance of \citet{Prokesova13}, we chose $R$ to be approximately equal to the range of interaction in the relevant dataset. The practical range of dependence is defined in the geostatistics literature to be the distance at which the spatial auto- or cross-correlation decays to 0.05. We calculated the practical range for each of our models, motivating our choice of $R=0.1$ for models 1 and 2, and $R=0.25$ for  models 3 and 4. 
In a small proportion of runs, the MPLE procedure returned seemingly degenerate estimates of either $\hat{\alpha}_{pq}$ or $\hat{\sigma}_{pq}$, $p=1,2$, with one or the other being returned equal to their upper bound. This was found to occur when the majority of the points in the corresponding dataset lay in the boundary region created using the above values of $R$. In this scenario, the number of points contributing to the Palm log-likelihoods \eqref{eq:bivpalmloglik}-\eqref{eq:palmloglik} is reduced, leading to a loss of accuracy in the MPLE procedure. We therefore counter this phenomenon by decreasing $R$ when necessary. When the initial attempt returns estimates of any of the Mat\'ern scale or power parameters greater than 95\% of their corresponding upper bound, we iteratively repeat the MPLE procedure, reducing $R$ by 0.01 each time, until all scale and power estimates are below this 95\% threshold. We found this to be an adequate, if somewhat ad-hoc remedy to the problem. After applying our iterative fix, for each of the four models considered, fewer than 8 of the 500 Monte Carlo runs returned any Mat\'ern scale or power estimates greater than $50\%$ of their corresponding upper bound.


In Table \ref{table:MCsims_fullMatern}, we provide summary statistics for the Monte Carlo sampling distributions of the parameters in Models 1-4. For the smoothness parameters, we report the modal estimate from our Monte Carlo simulations, as we consider only three potential values for these parameters. For the estimated scales of anisotropy, we provide the median of the Monte Carlo samples, along with the sample standard deviation, since their sampling distributions display evidence of skewness. For the estimated angles of anisotropy, as well as the Mat\'ern scale and power parameters, we provide the MC sample mean and the MC sample standard deviation. 
The sampling distributions of the parameter estimates for Model 1 are depicted in Figure \ref{fig:MCsims_Model1}, and the corresponding figures for Models 2-4 are provided in the Appendix.

\begin{table}
	\begin{center}
		\begin{tabular}{|l|ccc|ccc||cc|}
			\hline
			&	$\theta_{11}$	&	$\theta_{22}$ 	&	$\theta_{12}$ 
			&	$\zeta_{11}$	&	$\zeta_{22}$ 	&	$\zeta_{12}$ 
			&	$\mu_{1}$		&	$\mu_2$	\\
			\hline
			Dataset 1 
			&	$36^{\circ}$	& 	$72^{\circ}$	&	 $54^{\circ}$ 	
			&		0.20		&		0.20		&		0.35	
			&		4.75		&		4.5			\\
			MC estimate
			&	$40.36^{\circ}$		&	$73.87^{\circ}$		&	$67.07^{\circ}$
			&	0.25	&	0.24	&	0.41
			&	3.14	&	3.06	\\
			MC std.~dev.
			&	$21.72^{\circ}$		&	$16.19^{\circ}$		&	$37.49^{\circ}$
			&	0.33	&	0.27	&	0.48
			&	1.62	&	2.31	\\
			\hline
			
			Dataset 2 
			&	$36^{\circ}$	& 	$72^{\circ}$	&	 $54^{\circ}$ 	
			&		0.40		&		0.40		&		0.60	
			&		4.75		&		4.5			\\
			MC estimate
			&	$42.51^{\circ}$		&	$75.38^{\circ}$		&	$70.90^{\circ}$
			&	0.41	&	0.41	&	0.58
			&	3.65	&	3.47	\\
			MC std.~dev.
			&	$26.54^{\circ}$		&	$22.22^{\circ}$		&	$43.54^{\circ}$
			&	0.35	&	0.28	&	0.44
			&	1.91	&	2.82	\\
			\hline
			
			Dataset 3 
			&	$36^{\circ}$	& 	$72^{\circ}$	&	 $54^{\circ}$ 	
			&		0.20		&		0.20		&		0.35	
			&		5.75		&		5.625		\\
			MC estimate
			& $32.49^{\circ}$	& $69.87^{\circ}$	&	$52.07^{\circ}$
			&	0.22	&	0.23	&	0.34
			&	2.95	&	2.81	\\
			MC std.~dev.
			& $14.49^{\circ}$	& $8.53^{\circ}$	& $27.22^{\circ}$
			&	0.08	&	0.12	&	0.29
			&	1.40	&	2.27	\\
			\hline
			
			Dataset 4 
			& $36^{\circ}$		&  $72^{\circ}$		& $54^{\circ}$ 	
			&		0.40		&		0.40		&		0.60	
			&		5.75		&		5.625		\\
			MC estimate
			&$39.73^{\circ}$	& $74.24^{\circ}$	& $64.94^{\circ}$
			&	0.39	&	0.40	&	0.52
			&	4.16	&	4.28	\\
			MC std.~dev.
			& $21.98^{\circ}$	& $18.31^{\circ}$	& $37.58^{\circ}$
			&	0.19	&	0.23	&	0.25
			&	1.72	&	1.36	\\
			\hline
		\end{tabular}

		\begin{tabular}{|l|ccc|ccc|ccc|}
			\hline
			&	$\alpha_{11}$	&	$\alpha_{22}$ 	&	$\alpha_{12}$ 	&	$\nu_{11}$		&	$\nu_{22}$ 		&	$\nu_{12}$ 		&	$\sigma_{11}$	&	$\sigma_{22}$ 	&	$\sigma_{12}$ 	\\
			\hline
			
			Dataset 1
			&		0.045	&		0.065	&		0.050
			&		0.5		&		0.5		&		0.5		
			&		4.00	&		4.50	&		1.97		\\
			MC estimate
			&	0.042	&	0.116	&	0.047
			&	0.5		&	0.5		&	0.5
			&	4.52	&	4.13	&	1.28		\\
			MC std.~dev.
			&	0.033	&	0.066	&	0.021
			&	-		&	-		&	-
			&	2.93	&	2.11	&	1.00	\\
			\hline
			
			Dataset 2
			&		0.045	&		0.065	&		0.050
			&		0.5		&		0.5		&		0.5		
			&		4.00	&		4.50	&		2.30		\\
			MC estimate
			&		0.084	&		0.138	&		0.049
			&		0.5		&		0.5		&		0.5		
			&		4.19	&		4.14	&		1.37		\\
			MC std.~dev.
			&		0.586	&		0.689	&		0.029
			&		-		&		-		&		-		
			&		2.47	&		2.50	&		0.98	\\
			\hline
			
			Dataset 3
			&		0.090	&		0.120	&		0.100 
			&		0.5		&		0.5		&		0.5		
			&		2.00	&		2.25	&		0.98		\\
			MC estimate
			&		0.147	&		0.174	&		0.110
			&		0.5		&		0.5		&		5.0		
			&		2.91	&		2.98	&		0.58	\\
			MC std.~dev.
			&		0.593	&		0.497	&		0.067
			&		-		&		-		&		-		
			&		2.36	&		2.40	&		0.93	\\
			\hline
			
			Dataset 4
			&		0.090	&		0.120	&		0.100 
			&		0.5		&		0.5		&		0.5		
			&		2.00	&		2.25	&		1.15		\\
			MC estimate
			&		0.191	&		0.181	&		0.111
			&		0.5		&		0.5		&		5.0		
			&		2.85	&		2.52	&		0.72	\\
			MC std.~dev.
			&		0.763	&		0.487	&		0.075
			&		-		&		-		&		-		
			&		2.19	&		1.77	&		1.02	\\
			\hline
		\end{tabular}
	\end{center}
	\caption{Monte Carlo estimates and standard errors for the anisotropy (top) and Mat\'ern (bottom) parameters in four distinct models. All estimates and errors are given to 2dp, apart from those for the scale parameters; these are presented to 3dp, due to the magnitude of the errors.
	}
	\label{table:MCsims_fullMatern}
\end{table}

\begin{figure}
	\includegraphics[width=\textwidth]{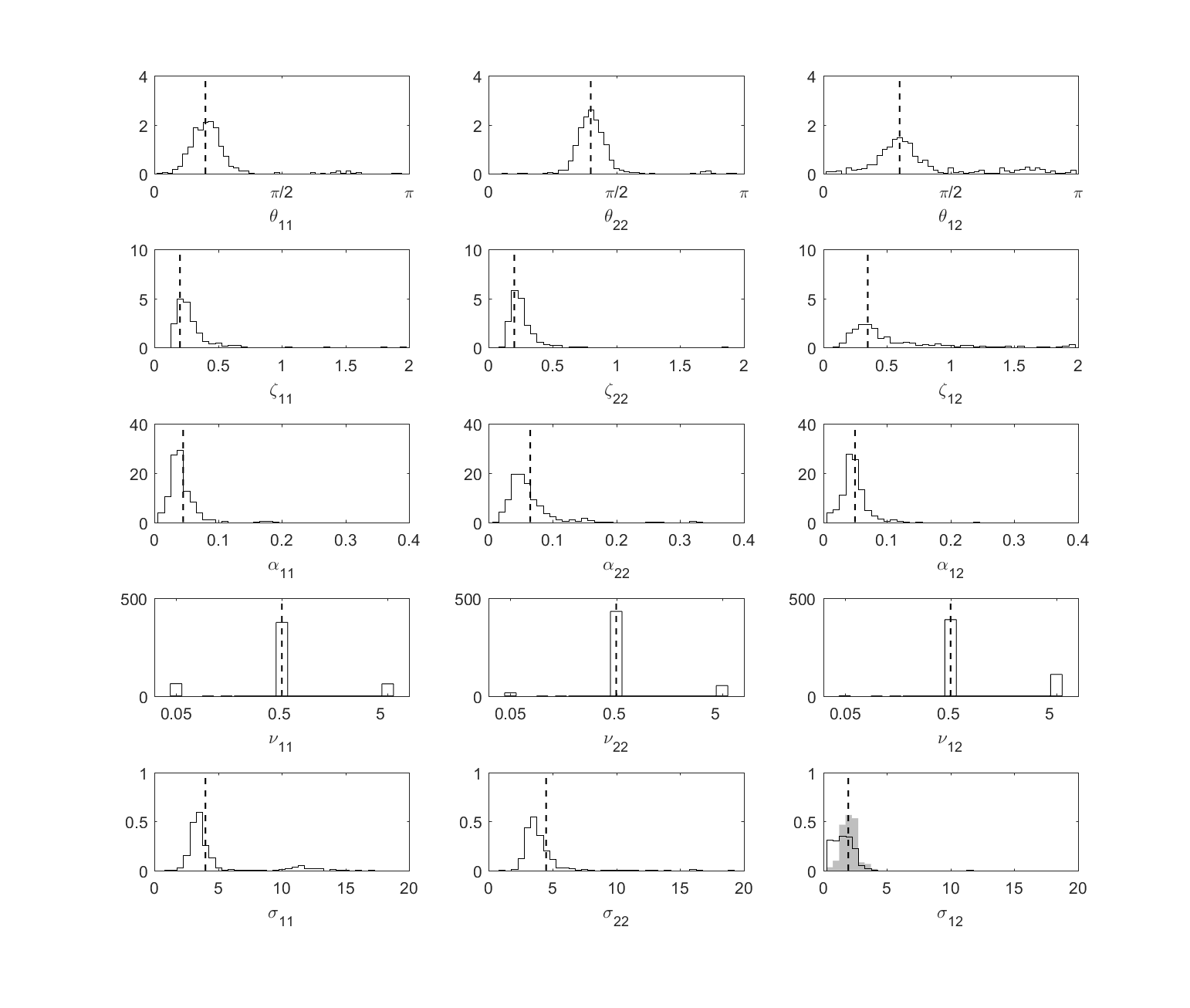}
	\vspace{-40pt}
	\caption{Histograms of the parameter distributions for the synthetic bivariate geometric anisotropic LGCP with Mat\'ern covariance structure specified by Model 1. The parameter values used to generate each dataset are marked by vertical dashed lines.}
	\label{fig:MCsims_Model1}
\end{figure}

From these results, we can identify some general conclusions regarding the performance of our model fitting procedure. Firstly, for all datasets, the estimated anisotropy parameters are in reasonable agreement with their corresponding true values. There is room for improvement in accuracy, especially in the estimated values of $\hat{\theta}_{pq}$, $p,q=1,2$; as noted above, this can be achieved through reducing the range of distances, $r$, over which we seek to characterise anisotropy. The broad accuracy of these estimates, however, suggests that our bivariate generalisation of \citeauthor{Rajala16}'s method of estimating anisotropy has been successful. 
 
Similarly, we can see that the Mat\'ern scale parameters, $\alpha_{pq}$, have been satisfactorily estimated for all four models. For models 1 and 2, we have also recovered the correct values of the Mat\'ern smoothness parameters $\nu_{pq}$. For models 3 and 4, however, there are some notable inaccuracies in estimating $\nu_{12}$. We attribute this to the lower power in the between-process dependence structures for models 3 and 4. For all four models, the power parameter estimates $\sigma_{pq}$, $p,q=1,2$ show reasonable accuracy, though there is consistent underestimation of the joint dependence power parameter. Further examination of the empirical distributions of $\hat{\sigma}_{12}$ suggests that this can be attributed to our use of constrained optimisation of the bivariate Palm log-likelihood.
In the final panel of Figures \ref{fig:MCsims_Model1} and \ref{fig:MCsims_Model2}-\ref{fig:MCsims_Model4}, we have overlain the empirical parameter distribution for $\hat{\sigma}_{12}$, restricted to those MC simulations where $\hat{\sigma}_{12}$ was not equal to the upper bound dictated by $\hat{\sigma}_{11}$ and $\hat{\sigma}_{22}$. This suggests that our use of constrained optimisation limits the accuracy of the estimated power parameter; this is the cost of ensuring that each fitted parameter vector specifies a valid multivariate dependence structure.

Finally we note that, across all models, the estimation of $\mu_1$ and $\mu_2$ is poor. We found that this can be improved by reducing the MPLE tuning parameter $R$, but with a loss of accuracy in the resulting covariance parameters. In practice, we can of course avoid this trade-off by instead using the classical estimator for the intensity, $\hat{\lambda}_p=n_p/|W|$, and combining this with $\hat{\sigma}_{pp}$ to obtain a more accurate estimate for $\mu_{p}$. 

Overall, these results indicate reasonable success for our model fitting procedure, and warrant its use in exploring the model's effectiveness in characterising real data. 

\subsection{Application to ecological data}
\label{subsec:BCIimplementation}

In order to demonstrate the utility of our multivariate geometric anisotropic framework, we fit our multivariate Mat\'ern geometric anisotropic LGCP to a bivariate point pattern from a 50ha plot in the BCI forest stand in Panama. Our point pattern of interest comprises two tree species, \textit{Cecropia obtusifolia} and \textit{Spondias radlkoferi}. To ease comparison with the studies in the previous section, we rescale the coordinates to the half-unit window $[0,1]\times[0,0.5]$; this rescaled bivariate point pattern is displayed in Figure \ref{fig:BCIdata}. \textit{C. obtusifolia} and \textit{S. radlkoferi} were chosen as 
a preliminary study of the data revealed empirical evidence of between-process anisotropy at a range of $r=50m$. This is demonstrated in Figure \ref{fig:MCsecKfun}, which we describe below. This preliminary evidence also motivates the scales over which we seek to characterise anisotropy in the data: for both the marginal and cross-dependence relationships, we estimate $\theta_{pq}$ using Fry processes consisting of only those point pairs separated by $r\in(0,0.05)$, and we estimate $\zeta_{pq}$ using $b_1=0$ and $b_2=0.05$ as the limits of integration in $\hat{V}_{pq,\hat{\theta}_{pq}}(\zeta)$.

As in Section \ref{subsec:proofofconcept}, we must also choose a value of the MPLE tuning parameter $R$. Once again, we do so by consulting the marginal and cross-pair correlation functions for the isotropised data, once the marginal and between-process anisotropy parameters have been estimated. We found that the corresponding implied practical range, both within-species and between-species, was in the interval $[0.1,0.2]$. We therefore executed the MPLE portion of our model fitting procedure for each of $R=0.1,0.15,0.2$.

\begin{figure}
	\includegraphics[width=\textwidth]{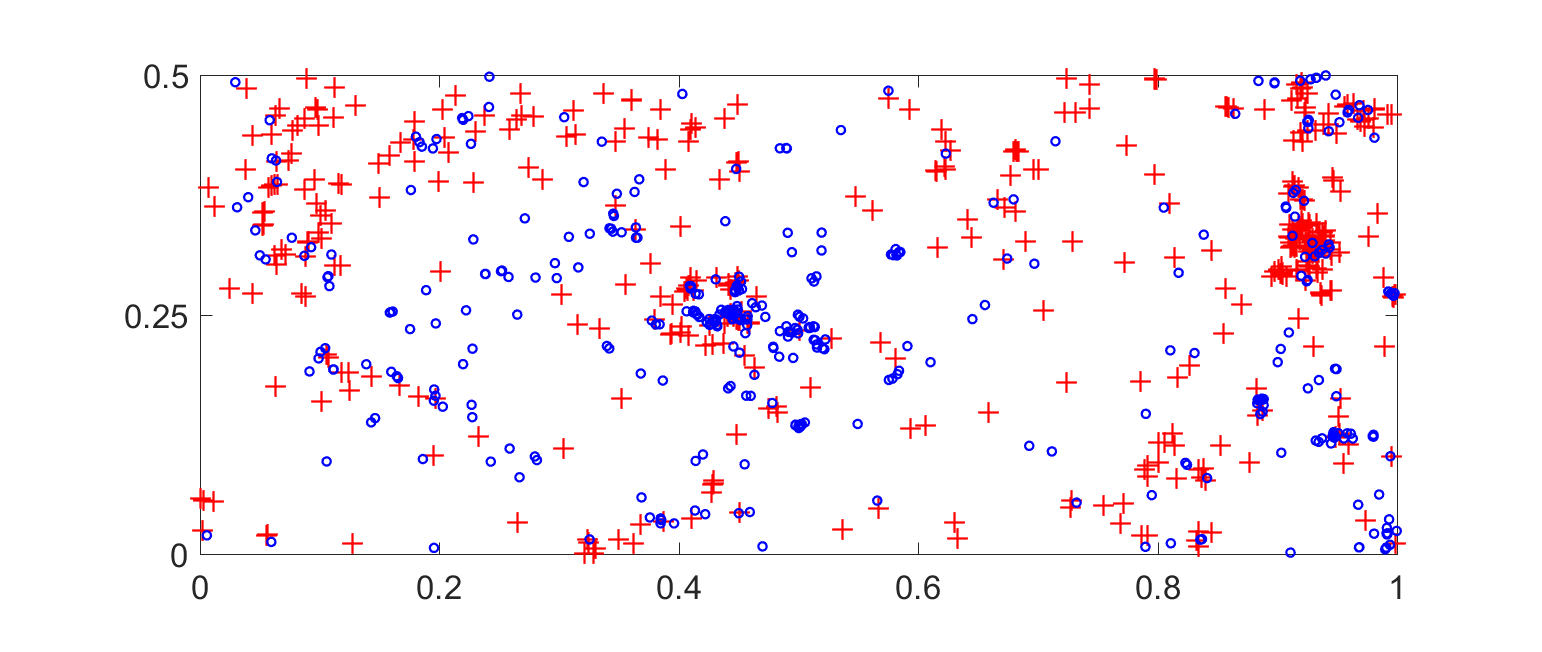}
	\caption{Rescaled point pattern data from the 50ha tropical rainforest census plot on Barro Colorado Island. Two species are shown: \textit{Cecropia obtusifolia} (blue circles) and \textit{Spondias radlkoferi} (red crosses).}
	\label{fig:BCIdata}
\end{figure}

For the proof-of-concept studies in Section \ref{subsec:proofofconcept}, we were able to avoid constraining the anisotropy parameters during the estimation procedure, as we knew that their true values satisfied the relevant model validity conditions of Section \ref{subsec:MVMaterncorr}. When fitting the model to observed data, however, we have no such assurance. Instead of introducing any new constraints on the anisotropy parameters here, we acknowledge this uncertainty by checking each fitted model against Conditions 1-4; all of the fitted models we present here were found to satisfy these validity conditions.
In Section \ref{subsec:proofofconcept}, we also found that estimating the Mat\'ern parameters via constrained optimisation can result in underestimation of the overall power in the between-process covariance. This occurs when the estimated value of $\sigma_{12}$ is equal to the upper bound specified by the marginal dependence structures. By calculating this upper bound explicitly, and comparing with $\hat{\sigma}_{12}$, we can therefore ascertain whether each fitted model accurately represents the between-species dependence structure; this is important, as it describes the interspecific interaction between individual trees in our dataset.

To begin with, we applied our model-fitting procedure as in Section \ref{subsec:proofofconcept}. This resulted in the model specified by the first three rows of parameter estimates in Table \ref{table:BCIfittedmodels}. Since $\hat{\sigma}_{12}=\sigma^{UB}_{12}$ in each of these specifications, we conclude that none of these fitted models accurately represent the interspecific interaction in our dataset. 
Motivated by the observation that distinct values of the marginal smoothness parameters lead to prohibitively small values of $\sigma^{UB}_{12}$, we next proceeded to fix the smoothness parameters, $\nu_{11}=\nu_{22}=\nu_{12}=0.5$, such that we sought to fit a geometric anisotropic bivariate Exponential covariance structure to our data. Crucially, all of the discussion from Sections \ref{sec:model} and \ref{sec:modelfitting} is valid for fixed values of $\nu_{pq}$, $p,q=1,2$. The resulting parameter estimates for this model are given in rows 4-6 of Table \ref{table:BCIfittedmodels}. 
In order to demonstrate the utility of our multivariate anisotropic framework, we also fit an isotropic version of the multivariate Exponential LGCP to the same data, for the purpose of comparison. In practice, we achieve this by fixing $\zeta_{pq}=1$ and $\theta_{pq}=0$ for $p,q=1,2$, and implementing the MPLE portion of the model fitting procedure as described above, using $R=0.1,0.15,0.2$. The resulting three sets of estimated scale and power parameters for this model are given in the bottom three rows of Table \ref{table:BCIfittedmodels}. As is shown in this table, the interspecific interaction is well-represented in only one of each of the anisotropic and isotropic fitted Exponential models. We henceforth restrict our attention to these two fitted models. 

\begin{table}[t!]
	\begin{center}
		\begin{tabular}{|ccc|ccc||cc|}
			\hline
			$\hat{\theta}_{11}$	&	$\hat{\theta}_{22}$ & $\hat{\theta}_{12}$ 
			&	$\hat{\zeta}_{11}$	&	$\hat{\zeta}_{22}$ 	& $\hat{\zeta}_{12}$ 
			&	$\hat{\lambda}_{1}$	&	$\hat{\lambda}_{2}$ 	\\ 
			\hline
			$158.89^{\circ}$	& 	$87.65^{\circ}$	&	 $127.39^{\circ}$
			&	0.51		&	0.39	&	0.53
			&	824		&	878		\\
			
			\hline
		\end{tabular}

		\begin{tabular}{|l|c|ccc|ccc|ccc|}
			\hline
			Covariance model 
			&	$R$
			&	$\hat{\alpha}_{11}$	&	$\hat{\alpha}_{22}$ 	& $\hat{\alpha}_{12}$ 
			&	$\hat{\nu}_{11}$	&	$\hat{\nu}_{22}$ 		& $\hat{\nu}_{12}$ 
			&	$\hat{\sigma}_{11}$	&	$\hat{\sigma}_{22}$ 	& $\hat{\sigma}_{12}$ \\
			\hline
			Anisotropic
			& 0.1
			&	0.03	&	0.71	&	0.15
			&	0.05	&	0.5		&	0.5
			&	11.55	&	13.66	&	1.16$^*$	\\
			Mat\'ern
			& 0.15
    		&	10.00	&	0.08	&	0.11
    		&	0.05	&	5.0		&	5.0
    		&	19.05	&	3.09	&	1.01e-07$^*$	\\
   			& 0.2
    		&	0.07	&	0.10	&	0.14
    		&	0.05	&	5.0		&	5.0
    		&	12.50	&	2.00	&	0.03$^*$	\\
			\hline
			Anisotropic
			& 0.1
			&	0.03	&	0.71	&	0.04
			&	-		&	-		&	-
			&	3.09	&	13.66	&	1.61$^*$	\\
			Exponential
			&\bf{0.15}
			&\bf{0.10}	&\bf{0.18}	&\bf{0.12}
			&	-		&	-		&	-
			&\bf{3.47}	&\bf{5.22}	&\bf{2.45}	
													\\
			& 0.2
			&	0.07	&	0.17	&	0.08
			&	-		&	-		&	-
			&	3.30	&	3.03	&	1.96$^*$	\\

			\hline
			Isotropic
			& 0.1
			&	0.03	&	0.75	&	0.04
			&	-		&	-		&	-	
			&	3.32	&	8.29	&	1.44$^*$	\\
			Exponential
			& 0.15
			&	0.08	&	0.13	&	0.08
			&	-		&	-		&	-		
			&	3.34	&	2.79	&	2.33$^*$	\\
			
			&\bf{0.2}
			&\bf{0.12}	&\bf{0.09}	&\bf{0.10}		
			&	-		&	-		&	-	
			&\bf{4.28}	&\bf{3.40}	&\bf{3.50} 
												\\
			\hline
		\end{tabular}
	\end{center}
	\caption{Parameter estimates for three bivariate LGCPs, fitted to the tropical rainforest data described in the text. The anisotropy parameters in the top table apply to both anisotropic models described in the bottom table. Those values of $\hat{\sigma}_{12}$ marked with an asterisk ($^*$) are equal to the corresponding upper bound $\sigma^{UB}_{12}$. The two models that we choose to assess using global envelope tests, are highlighted in bold.}
	\label{table:BCIfittedmodels}
\end{table}

In order to assess different aspects of each model's performance, we use global envelope tests (GETs), in which we compare second-order statistics for the observed data with those of $M$ bivariate point patterns, independently simulated from the fitted model. Such tests were developed by \citet{Myllymaki17} to address multiple testing concerns with regards to the popular use of Monte Carlo envelope tests \citep{Loosemore06,Baddeley14}. The envelopes provided by the GETs describe a proper statistical test: if the observed test statistic lies outside the simulated envelope \textit{at any instance}, then the null hypothesis that the observed data belong to the fitted model may be rejected. In order to construct the envelopes, we use one of the following two approaches, both of which are described in detail by \citet{Myllymaki17}. For symmetric second-order statistics, we use the scaled studentized maximum absolute difference (MAD) to construct the critical bounds. For asymmetric second-order statistics, it is more appropriate to construct the envelopes using the scaled directional quantile MAD. In both cases, we configure our tests such that they have a global type I error probability of 0.1, using $M=499$.

To assess each model's description of bivariate anisotropy in the data, we estimate the the sector-$K$-function \eqref{eq:sectorKfun} at a distance $r=0.05$, and at the angles $\phi=k\pi/60$, $k=0,\ldots,60$. This assessment is summarised in Figure \ref{fig:MCsecKfun}, which gives the estimated marginal and between-process sector-$K$-functions for the observed BCI data, along with their corresponding directional quantile MAD envelopes.
To assess the fitted model's ability to replicate aggregation in the observed multi-type point pattern, we use the `$p$-to-$q$' nearest-neighbour distance distribution function $G_{pq}(r)$ \citep{VanLieshout99}. This describes the empirical distribution of the absolute distance between the typical point of type $p$ and its nearest point of type $q$ and is, in general, asymmetric in $p,q$. 
We calculate $G_{pq}(r)$ at a range of distances, using a bivariate version of the border-corrected estimators detailed by \citet[][\S 8.11.3]{Baddeley15}. The resulting estimates for the observed BCI data are provided in Figure \ref{fig:MCnnddfun}
, along with their corresponding studentized MAD envelopes.

From Figure \ref{fig:MCsecKfun}, we can see that the two chosen species in the BCI forest stand exhibit anisotropic interspecific interaction at a range of $50m$: for $\phi\in[7\pi/60,9\pi/60]\cup[12\pi/60,18\pi/60]$, the estimated sector-$K$-function $K_{12}^a(0.05,\phi)$ lies outside of the envelope generated by a multivariate isotropic LGCP. The $p$-value for the global envelope test that was carried out for this statistic was $0.058$, indicating departure from the isotropic model when using a global type I error probability of 0.1. From the bottom row of panels, we can see that our multivariate geometric anisotropic LGCP can comfortably replicate this observed heterogeneity.  

Finally, the bottom row of panels in Figure \ref{fig:MCnnddfun} demonstrates that our fitted multivariate anisotropic LGCP also accurately captures the clustering behaviour exhibited by the observed bivariate data. The top panels in this Figure suggest that, despite not being able to account for the anisotropy in the data, the multivariate isotropic LGCP has also captured the clustering behaviour evident in this particular example.

\begin{figure}[t!]
\includegraphics[width=\textwidth]{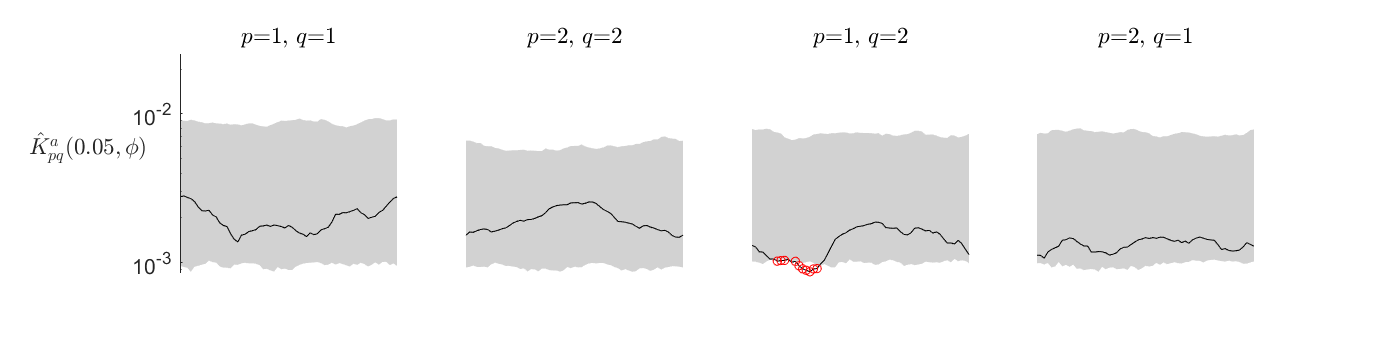}
\vspace{-30pt}
\\
\includegraphics[width=\textwidth]{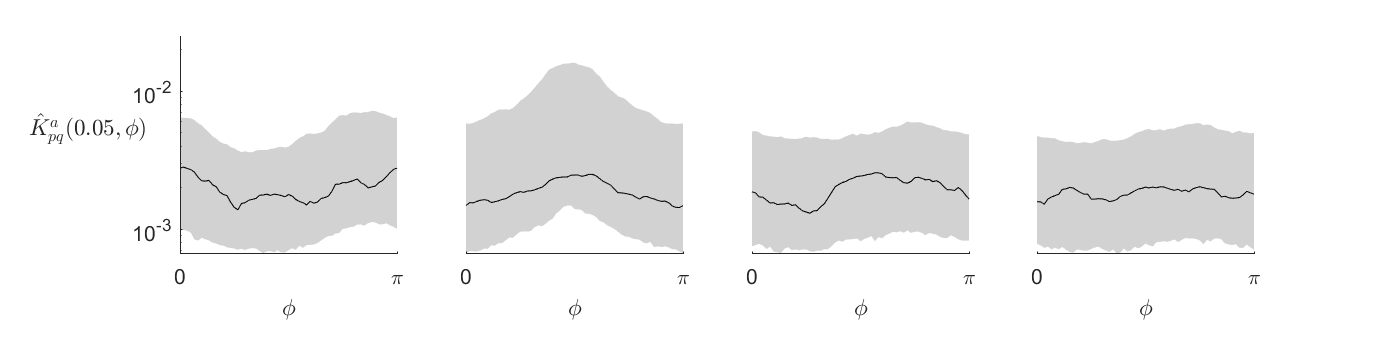}
\caption{Estimates of the sector-$K$-function \eqref{eq:sectorKfun}, at a fixed range of $50m$, for the observed bivariate BCI point pattern (black line), along with the corresponding 90\% directional-quantile MAD envelopes obtained from a fitted multivariate geometric anisotropic LGCP (bottom row) and from a fitted multivariate isotropic LGCP (top row). Departure of the data from the adopted model is highlighted with red circles. The vertical axes are presented on a log-scale to highlight the departure of the data from the isotropic model in the third panel.
} 
\label{fig:MCsecKfun}
\end{figure}
\begin{figure}[t!]
	\includegraphics[width=\textwidth]{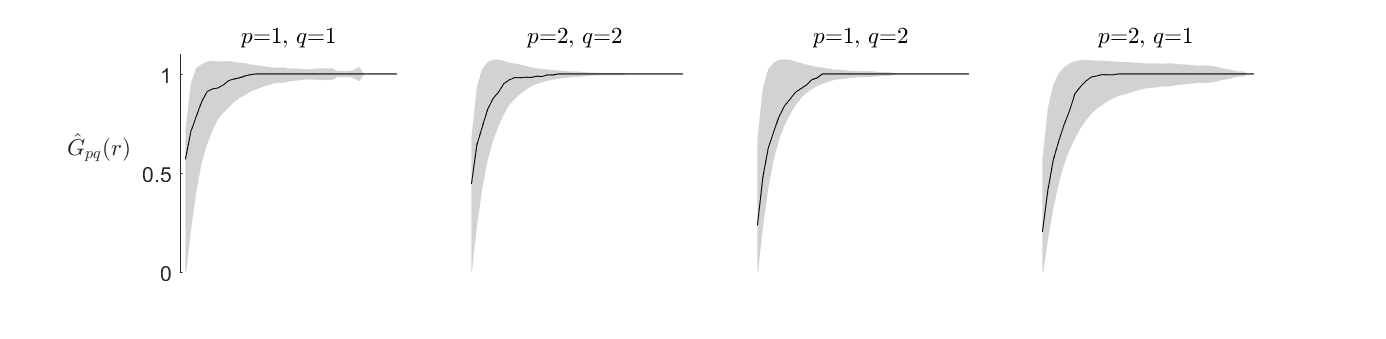}
	\vspace{-30pt}
	\\
	\includegraphics[width=\textwidth]{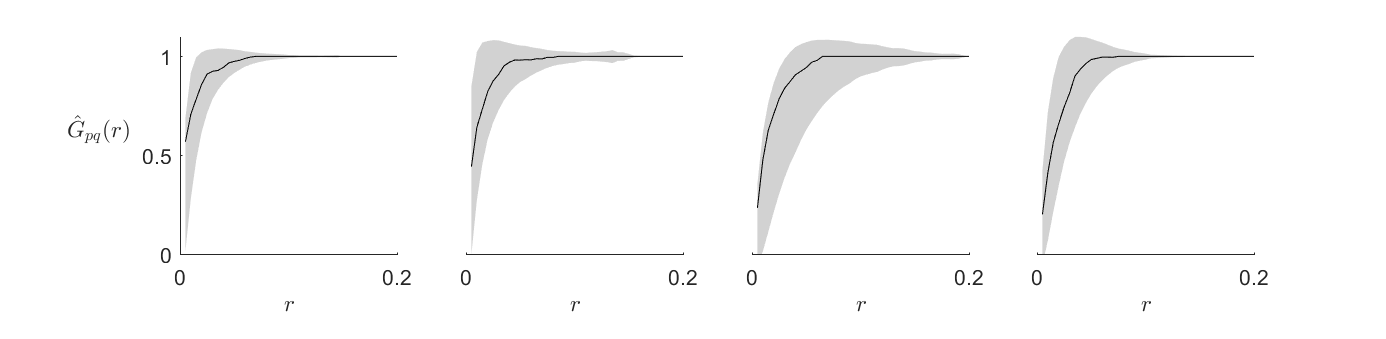}
	\caption{Estimated nearest-neighbour distance distribution functions for the observed bivariate BCI point pattern (black line), along with the corresponding 90\% studentized MAD envelopes obtained from a fitted multivariate geometric anisotropic LGCP (bottom row) and from a fitted multivariate isotropic LGCP (top row).} 
	\label{fig:MCnnddfun}
\end{figure}

\section{Discussion}


Using the model-fitting methodology described in Section \ref{sec:modelfitting}, we have shown that by incorporating geometric anisotropy into the between-process dependence, as well as the marginal dependence, we can construct a LGCP that more accurately replicates any rotationally heterogeneous interaction between points in a multi-type point pattern. 
We have focussed here on a covariate-free approach, motivated in part by the desire to allow the description of anisotropic between-process dependence in data for which there are no explanatory spatial variables. Nevertheless, the models presented here are flexible enough to use (potentially incomplete) covariate information where it is available. 
Indeed, an interesting first extension of this work would be to incorporate covariates into the first-order description of the GRF underlying our LGCPs; for instance, the expected value of the GRF could be specified through a linear regression model, and inference with respect to the regression parameters may be achievable through the use of estimating functions \citep[e.g.][]{Waagepetersen08,Waagepetersen09}. Such an approach would allow the user to exploit any knowledge of spatial covariates whilst being confident that any residual heterogeneity in the data would be accounted for by the increased flexibility of the multivariate geometric anisotropic second-order dependence structure.

\section{Acknowledgements}
The work of J. S. Martin, D. J. Murrell and S. C. Olhede was supported by the UK Engineering and Physical Sciences Research Council via EP/N007336/1, and EP/L001519/1. S. C. Olhede also acknowledges support from the 7th European Community Framework Programme via a Grant CoG 2015- 682172NETS (Olhede). 

The BCI forest dynamics research project was founded by S. P. Hubbell and R. B. Foster and is now managed by R. Condit, S. Lao, and R. Perez under the Center for Tropical Forest Science and the Smithsonian Tropical Research in Panama. Numerous organizations have provided funding, principally the U.S. National Science Foundation, and hundreds of field workers have contributed.

\newpage
\appendix

\section{Appendix}

\subsection{Proof of Proposition \ref{prop:mvgaMatern}}
\label{appx:proofprop1}
In Proposition \ref{prop:mvgaMatern}, we state that Conditions 1-4 are sufficient for the geometric anisotropic Mat\'ern function in \eqref{eq:mvgaMatern} to specify a valid multivariate covariance model, and we sketch the proof here. This proof is similar to that of Theorem 1 of \citet{Apanasovich12}, with additional consideration required to account for geometric anisotropy. As such, our proof depends on the following lemmas, due to \citet{Apanasovich12}, proofs for which can be found in that paper.
\begin{lemma}{ \citep{Apanasovich12}. }
Let $0<b_p<\infty$, $p=1,\ldots,P$, $\delta\ge0$, and $B_{pq}>0$, $p,q=1,\ldots,P$, be such that the matrix $(-B_{pq})_{p,q=1}^P$ is conditionally nonnegative definite. Then the $P\times P$ matrix with entries
$$
\frac{\Gamma(b_p + b_q + \delta)}{B_{pq}^{b_p+b_q+\delta}} \qquad p, q = 1,\ldots P,
$$
is nonnegative definite.
\end{lemma}
\begin{lemma}{ \citep{Apanasovich12}. }
Let $\delta\ge0$ and $B_{pq}$, $p,q=1,\ldots,P$ be as in Lemma 1. Then the matrix with $(p,q)^{\textrm{th}}$ entry
$$
\frac{1}{(B_{pq} + \delta )^r} \qquad p, q = 1,\ldots,P,
$$
is nonnegative definite, for any $0<r<\infty$.
\end{lemma}

\begin{proof}[of Proposition \ref{prop:mvgaMatern}]
We operate in the spectral domain: by Cram\'er's generalisation of Bochner's Theorem \citep{Cramer45}, the covariance matrix $\left(C_{pq}(h)\right)_{p,q=1}^{P}$ is nonnegative definite if and only if the corresponding matrix of spectral densities $\left(f_{pq}(\omega)\right)_{p,q=1}^P$ is also nonnegative definite. We therefore consider the form of the multivariate spectral density function, corresponding to \eqref{eq:mvgaMatern}: 
\begin{eqnarray*}
f_{pq}(\omega) & = & \left|\Sigma_{pq}\right|^{1/2} f_{I,pq}\left(\Sigma_{pq}^{1/2}\omega\right) 
\\
	& = & \frac{|\Sigma_{pq}|^{1/2}\sigma_{pq} \Gamma(\nu_{pq}+d/2)}{\pi^{d/2}\Gamma(\nu_{pq})}		\left( \frac{ 4\nu_{pq}}{\alpha_{pq}^2}\right)^{\nu_{pq}} \left( \frac{4\nu_{pq}}{\alpha_{pq}^2} + \|\Sigma_{pq}^{1/2}\omega\|^2 \right)^{-\nu_{pq}-d/2},
\end{eqnarray*}
where each anisotropic deformation matrix $\Sigma_{pq}$ is defined according to \eqref{eq:geoanisotSigmadef} in terms of $\theta_{pq}$ and $\zeta_{pq}$. We can decompose this spectrum as follows, in the process defining four terms numbered I to IV:
\begin{eqnarray}
\label{eq:mvgaMaternspecdecomp}
f_{pq}(\omega) & = 
	  & \overbrace{\frac{\Gamma(\frac{\nu_{pp}+\nu_{qq}}{2}+\frac{d}{2})}{ \left( \frac{4\nu_{pq}}{\alpha_{pq}^2} + \|\Sigma_{pq}^{1/2}\omega\|^2 \right)^{\frac{\nu_{pp}+\nu_{qq}}{2}+\frac{d}{2}} }}^{\textrm{Term I}}
	\times\overbrace{ \left(\frac{\frac{4\nu_{pq}}{ \alpha_{pq}^2} }{\frac{4\nu_{pq}}{\alpha_{pq}^2} + \|\Sigma_{pq}^{1/2}\omega\|^2}\right)^{-\Delta_{\nu} A_{\nu,pq}}}^{\textrm{Term II}}	\nonumber\\
	&&	
	\hspace{-20pt}
	\times \underbrace{\frac{1}{\left(\frac{4\nu_{pq}}{\alpha_{pq}^2} + \|\Sigma_{pq}^{1/2}\omega\|^2\right)^{\Delta_{\nu}}}}_{\textrm{Term III}}
	\times \underbrace{\frac{|\Sigma_{pq}|^{1/2}\sigma_{pq}\Gamma(\nu_{pq}+d/2)}{\pi^{d/2}\Gamma\left(\frac{\nu_{pp}+\nu_{qq}}{2}+\frac{d}{2}\right)\Gamma(\nu_{pq})}\left(\frac{4\nu_{pq}}{\alpha_{pq}^2}\right)^{\Delta_{\nu}+\frac{\nu_{pp}+\nu_{qq}}{2}}}_{\textrm{Term IV}},
\end{eqnarray}
where $A_{\nu,pq} = 1-\left\{\nu_{pq}-\left(\nu_{pp}+\nu_{qq}\right)/2\right\}/\Delta_{\nu}$ is the $(p,q)$-element of a valid nonnegative correlation matrix; nonnegative definiteness of the spectral matrix follows from nonnegative definiteness of the matrices formed from these constituent terms. 

Condition \ref{cond:alpha} is sufficient to guarantee nonnegative definiteness of the matrices with elements given by either the first or third terms in \eqref{eq:mvgaMaternspecdecomp}; this can be seen for the former by applying Lemma 1 and for the latter by applying Lemma 2. 

Conditions \ref{cond:nu}, \ref{cond:alpha} and \ref{cond:Sigma} are sufficient to guarantee nonnegative definiteness of the matrix with elements given by the second term of \eqref{eq:mvgaMaternspecdecomp}. To see this, we first rewrite the second term in \eqref{eq:mvgaMaternspecdecomp} as
\begin{eqnarray*}
\left(\frac{\frac{4\nu_{pq}}{ \alpha_{pq}^2} }{\frac{4\nu_{pq}}{\alpha_{pq}^2} + \|\Sigma_{pq}^{1/2}\omega\|^2}\right)^{-\Delta_{\nu} A_{\nu,pq}} 
& = & 
\exp\left\{\Delta_{\nu} A_{\nu,pq}\left[-\log\left(1-\frac{\|\Sigma_{pq}^{1/2}\omega\|^2}{\frac{4\nu_{pq}}{\alpha_{pq}^2} + \|\Sigma_{pq}^{1/2}\omega\|^2}\right)\right]\right\}\\
& = & 
\prod_{k=1}^{\infty}
\exp\left\{\frac{\Delta_{\nu} A_{\nu,pq}}{k}\left[\frac{\|\Sigma_{pq}^{1/2}\omega\|^2}{\frac{4\nu_{pq}}{\alpha_{pq}^2} + \|\Sigma_{pq}^{1/2}\omega\|^2}\right]^k\right\},\\
\end{eqnarray*}
where we note that the infinite expansion of the logarithm is valid when
$$
\frac{\|\Sigma_{pq}^{1/2}\omega\|^2}{\frac{4\nu_{pq}}{\alpha_{pq}^2} + \|\Sigma_{pq}^{1/2}\omega\|^2}<1,
$$
and this is satisfied at all times, since $4\nu_{pq}/\alpha_{pq}^2>0$.

Now, consider the matrices $B$ and $C$ with elements $b_{pq}\ge0$, $c_{pq}\ge0$, $p,q=1,\ldots,P$ and suppose that both $-B$ and $-C$ are conditionally nonnegative definite. By applying Lemma 2 (with $\delta=0$, $r=1$), we have that the matrix with elements $1/b_{pq}$ is nonnegative definite, and therefore by the Schur product theorem, we have that the matrix with elements $-c_{pq}/b_{pq}$ is conditionally nonnegative definite. Now, using the matrices in Conditions 2 and 4 in place of the matrices $-C$ and $-B$, respectively, we can state that the matrix with elements 
$$
-\frac{4\nu_{pq}/\alpha_{pq}^2}{\|\Sigma_{pq}^{1/2}\|^2}
$$
is conditionally nonnegative definite. By applying Lemma 2 once more (this time with $\delta=1$), we therefore have that the matrix with elements
$$
\left(\frac{4\nu_{pq}/\alpha_{pq}^2}{\|\Sigma_{pq}^{1/2}\|^2} + 1\right)^{-r}
$$
is nonnegative definite for all $r>0$. It is now clear that, since $A_\nu$ is nonnegative definite and $\Delta_\nu\ge0$ (both by Condition 1), each exponential argument within the product above specifies a nonnegative definite matrix. Repeated further use of the Schur product theorem therefore allows us to conclude that the matrix with elements given by the second term in \eqref{eq:mvgaMaternspecdecomp} is indeed nonnegative definite. 

Finally, Condition \ref{cond:sigma} states the nonnegative definiteness of the matrix with entries specified by the fourth term of \eqref{eq:mvgaMaternspecdecomp}, and so we may conclude the stated result.
\end{proof}

\subsection{Proofs of Remarks \ref{rmk:alpha}-\ref{rmk:Sigma}}
\label{appx:proofrmks}
In Remarks \ref{rmk:alpha}-\ref{rmk:Sigma}, we provide definitions of the correlation length, smoothness parameter and spatial deformation matrix for the geometric anisotropic Mat\'ern cross-covariance function $C_{pq}\left(h\left|\alpha_{pq},\nu_{pq},\sigma_{pq},\Sigma_{pq}\right.\right)$, in terms of the corresponding marginal quantities. In this subsection, we prove that these definitions satisfy Conditions \ref{cond:alpha}-\ref{cond:Sigma}, respectively. 

Recall that a matrix $A\in \mathbb{C}^{P}\times \mathbb{C}^{P}$ is conditionally nonnegative definite if, for all $x\in\mathbb{C}^P$ such that $\sum_{p=1}^Px_p=0$, $\sum_{p,q=1}^P x_p A_{pq}x_q^*\ge0$, where $x_p^*$ is the complex conjugate of $x_p$.

\begin{proof}[Proof of Remark \ref{rmk:alpha}]
This proof is given in the appendix of \citet{Apanasovich12} for a different Mat\'ern parameterisation; we translate it to the current Mat\'ern parameterisation here. Suppose that
\begin{equation}
\label{eq:xcovcorrlen_app}
\frac{4\nu_{pq}}{\alpha_{pq}^2} = \frac12\left(\frac{4\nu_{pp}}{\alpha_{pp}^2}+\frac{4\nu_{qq}}{\alpha_{qq}^2}\right) + \Delta_{\alpha}\left(1-A_{\alpha,pq}\right),
\qquad
p,q = 1,\ldots,P,
\end{equation}
with $\Delta_{\alpha}\ge0$ and $0\le A_{\alpha,pq}\le1$ that form a valid correlation matrix. Consider $x\in\mathbb{C}^{P}$ such that $\sum_{p=1}^Px_p = 0$. Using \eqref{eq:xcovcorrlen_app},
\begin{eqnarray*}
\sum_{p,q} x_p \frac{4\nu_{pq}}{\alpha_{pq}^2} x^*_q & = & 
	\frac12 \left\{
			\left(\sum_{p}x_p\frac{4\nu_{pp}}{\alpha_{pp}^2}\right)\left(\sum_{q}x_q^*\right)
		+	\left(\sum_{p}x_p\right)\left(\sum_{q}\frac{4\nu_{qq}}{\alpha_{qq}^2}x_q^*\right)
		\right\}
	\\
&&\qquad\qquad\qquad\qquad\qquad\qquad
	+	\Delta_{\alpha}\sum_{p}x_p\sum_{q}x_q^*
	- 	\Delta_{\alpha}\sum_{pq}x_p A_{\alpha,pq} x_q^*
	\\
& = & - \Delta_{\alpha}\sum_{pq}x_p A_{\alpha,pq} x_q^*
\le  0,\mbox{ as }A_{\alpha,pq}\mbox{ is nonnegative definite.}
\end{eqnarray*}
Hence, the matrix with $(p,q)$-element $-\frac{4\nu_{pq}}{\alpha_{pq}^2}$ is conditionally nonnegative definite.
\end{proof}

\begin{proof}[Proof of Remark \ref{rmk:sigma}]
Through straightforward manipulation of the expression in Remark \ref{rmk:sigma}, we see that the matrix in Condition \ref{cond:sigma} is equal to the matrix with $(p,q)$-element given by $V_pV_qA_{\sigma,pq}$, where $V_p$, $V_q$ and $A_{\sigma}$ are defined in Remark \ref{rmk:sigma}. Since $A_{\sigma}$ is a (nonnegative definite) correlation matrix, and since $V_p,V_q\ge0$, this is also nonnegative definite.
\end{proof}

\begin{proof}[Proof of Remark \ref{rmk:Sigma}]

We also give motivation for the chosen construction of $\Sigma_{pq}$. We wish to have $\Sigma_{pq}$ such that the $P\times P$ matrix with $(p,q)$-element $-\omega^T \Sigma_{pq} \omega$ is conditionally nonnegative definite. Now, for a $P\times P$ matrix with $(p,q)$-element $-C_{pq}$ to be nonnegative definite, a necessary condition is for 
$$
C_{pq} \ge \frac12\left(C_{pp}+C_{qq}\right), \qquad p,q=1,\ldots,P.
$$
It therefore follows that for Condition \ref{cond:Sigma} to hold, we need
$$
\omega^T\Sigma_{pq}\omega \ge \frac12\left(\omega^T\Sigma_{pp}\omega + \omega^T\Sigma_{qq}\omega\right), \qquad \forall \omega\in\mathbb{R}^2.
$$
Since this must hold for all $\omega\in\mathbb{R}^2$, we can consider the particular case for $\left\{\omega\in\mathbb{R}^2:\omega_2=0\right\}$, from which we can deduce
$$
\left[\Sigma_{pq}\right]_{11} \ge \frac12  \Big(\left[\Sigma_{pp}\right]_{11} +\left[\Sigma_{qq}\right]_{11}\Big),
$$
and similarly, we can deduce
$$
\left[\Sigma_{pq}\right]_{22} \ge \frac12  \Big(\left[\Sigma_{pp}\right]_{22} +\left[\Sigma_{qq}\right]_{22}\Big);
$$
this motivates the construction of the diagonal elements of $\Sigma_{pq}$ in Remark \ref{rmk:Sigma}: 
\begin{eqnarray*}
\left[\Sigma_{pq}\right]_{ii} = \frac12\left[\Sigma_{pp}+\Sigma_{qq}\right]_{ii} + \Delta^{(i)}_{\Sigma}\left(1-A^{(i)}_{\Sigma,pq}\right),\qquad i=1,2,
\end{eqnarray*}
where each $A_{\Sigma}^{(i)}$ is a $P\times P$ correlation matrix and each $\Delta_{\Sigma}^{(i)}$ is a nonnegative constant.
Now, consider $x\in\mathbb{C}^{P}$ such that $\sum_{p=1}^Px_p = 0$. We wish to show that
\begin{eqnarray*}
\sum_{p,q} x_p \left(\omega^T\Sigma_{pq}\omega\right) x_q^* = \omega^T \left(\sum_{p,q} x_p \Sigma_{pq} x_q^*\right) \omega \le 0\qquad \forall \omega\in\mathbb{R}^2.
\end{eqnarray*}

By expanding the above quadratic in $\omega$, and then substituting our chosen construction for the diagonal elements, we can simplify to obtain
\begin{eqnarray*}
 \omega^T \left(\sum_{p,q} x_p \Sigma_{pq} x_q^*\right) \omega & = & 
-\Delta_{\Sigma}^{(1)}B_{\Sigma}^{(1)}\omega_1^2 -\Delta_{\Sigma}^{(2)}B_{\Sigma}^{(2)} \omega_2^2  + 2\omega_1\omega_2  \left(\sum_{p,q} x_p \left[\Sigma_{pq}\right]_{12} x_q^*\right),
\end{eqnarray*}
where, for $i=1,2$, $B^{(i)}_{\Sigma}=\left(\sum_{pq}x_p A_{\Sigma,pq}^{(i)} x_q^*\right)$ is nonnegative, as $A^{(i)}_{\Sigma}$ is a correlation matrix. In order for this quadratic term to maintain the same sign for all $\omega\in\mathbb{R}^2$, we must be able to factorise it further, i.e.~we must be able to write
$$
 \omega^T \left(\sum_{p,q} x_p \Sigma_{pq} x_q^*\right) \omega = k_1\left(\omega_1\pm k_2\omega_2\right)^2.
$$
for some $k_1,k_2\in\mathbb{R}$. By expanding and equating terms, it is straightforward to show that this form can be obtained: we can write
\begin{eqnarray*}
\omega^T \left(\sum_{p,q} x_p \Sigma_{pq} x_q^*\right) \omega & = & -\Delta_{\Sigma}^{(1)}B_{\Sigma}^{(1)}\left(\omega_1\pm \omega_2\sqrt{\frac{\Delta^{(2)}_{\Sigma}B_{\Sigma}^{(2)}}{\Delta^{(1)}_{\Sigma}B_{\Sigma}^{(1)}}}\right)^2 
\end{eqnarray*}
iff the off-diagonal elements of $\Sigma_{pq}$ satisfy the relationship
\begin{eqnarray}
\label{eq:SigDiagOffdiagRelation}
\left(\sum_{p,q=1}^P x_p \left[\Sigma_{pq}\right]_{12}x_q^* \right)^2  =
	\Delta^{(1)}_{\Sigma}B_{\Sigma}^{(1)}
	\Delta^{(2)}_{\Sigma}B_{\Sigma}^{(2)}.
\end{eqnarray}
Note that this specifies a relationship between the diagonal and off-diagonal elements of the set of matrices $\{\Sigma_{pq}, p,q=1,\ldots,P\}$, which must be satisfied in order for the $P\times P$ matrix with $(p,q)$-element $-\omega^T \Sigma_{pq} \omega$ to be conditionally nonnegative definite. 

Note that, since each $\Sigma_{pq}$ is a deformation matrix with form given by \eqref{eq:geoanisotSigmadef}, its diagonal and off-diagonal elements must be consistent with the same choice of $(\theta_{pq},\zeta_{pq})$. This places a fundamental restriction on the form of each $\Sigma_{pq}$, which will, in general, not agree with the constraint in \eqref{eq:SigDiagOffdiagRelation}. We can circumvent this apparent incompatibility of restrictions on the set of deformation matrices $\{\Sigma_{pq}, p,q=1,\ldots,P\}$ by writing the off-diagonal elements in the form
\begin{equation}
\label{eq:SigOffdiagDecomp}
\left[\Sigma_{pq}\right]_{12} = \left[\Sigma_{pq}\right]_{21} = b_p + c_q + A_{\Sigma,pq}^{(3)},
\end{equation}
where $b,c\in\mathbb{R}^P$ are constant $P$-length vectors, $\Delta^{(3)}_{\Sigma}$ is a nonnegative constant, and $A_{\Sigma}^{(3)}\in\mathbb{R}^{P\times P}$ is a $P\times P$ real matrix that satisfies
\begin{equation*}
\left(\sum_{p,q=1}^P x_p A_{\Sigma,pq}^{(3)} x_q^*\right)^2 =  
	\Delta^{(1)}_{\Sigma}B_{\Sigma}^{(1)}
	\Delta^{(2)}_{\Sigma}B_{\Sigma}^{(2)}.
\end{equation*}
By specifying the off-diagonal elements of $\Sigma_{pq}$ in this way, we have that our conditional nonnegative definiteness restriction \eqref{eq:SigDiagOffdiagRelation} reduces to a restriction on $A_{\Sigma}^{(3)}$, which is unaffected by the need for $\Sigma_{pq}$ to maintain the form of a valid deformation matrix, specified by \eqref{eq:geoanisotSigmadef}; since there are no further restrictions on the form of $A_{\Sigma}^{(3)}$, such a matrix will certainly exist. 

Therefore, if the diagonal elements of the deformation matrix $\Sigma_{pq}$ are specified as in Remark \ref{rmk:Sigma}, the resulting off-diagonal elements (which are immediately specified via \eqref{eq:geoanisotSigmadef}) will always satisfy a valid decomposition \eqref{eq:SigOffdiagDecomp}, guaranteeing satisfaction of the relationship \eqref{eq:SigDiagOffdiagRelation} will be satisfied. This allows us to conclude that, if the diagonal elements of the deformation matrix $\Sigma_{pq}$ are specified as in Remark \ref{rmk:Sigma}, the $P\times P$ matrix with $(p,q)$-element $-\omega^T \Sigma_{pq} \omega$ will be conditionally nonnegative definite. 
\end{proof}

\subsection{Estimators of second-order summary statistics}
\label{appx:2ordestimators}
We present details of two estimators of second-order summary statistics that are used in our parameter estimation procedure. The first estimator we consider is for the isotropic cross-pair correlation function $g_{0,pq}(r)$, used in initialising the Mat\'ern power and scale parameters: 
\begin{equation}
\label{eq:isoxpcfest}
\hat{g}_{0,pq}(r) = \sum_{\substack{x_p\in X_p\cap W \\ x_q\in X_q\cap W}}^{\ne}\frac{\kappa_{h_r}(\|x_p-x_q\| - r)}{2 \pi r 
	\hat{\lambda}_p \hat{\lambda}_q
	|W\cap W_{x_p-x_q}|},
\end{equation}
where $\kappa_{h_r}$ is a radial kernel function with bandwidth $h_r$, $\hat{\lambda}_{p}$ is an estimator for the constant expected intensity of component $X_p$, defined in \eqref{eq:lamp_mup}, and $|W\cap W_{u}|$ is an edge correction factor, defined as the area of overlap between the observation window $W$ and its translation by $u\in\mathbb{R}^2$; without such a correction, due to the finite observation region, the estimator would underestimate the number of point pairs that lie within distance $r$ of each other. The use of this edge correction also renders $\hat{g}_{0,pq}(r) \ne \hat{g}_{0,qp}(r)$ in general. In \eqref{eq:isoxpcfest}, and in the remainder of the paper, the notation $\Sigma^{\neq}$ indicates summation over all point pairs formed of distinct points; for bivariate definitions such as \eqref{eq:isoxpcfest}, this is clearly only relevant for the case where $p=q$. For component $p$ of our multivariate LGCP, we choose to estimate the expected intensity parameter $\hat{\lambda}_p$ using the classical global intensity estimator, $\hat{\lambda}_p=n_p/|W|$. The choice of kernel function $\kappa_{h_r}$ is discussed by \citet{Illian08} and common choices include the Epanechnikov kernel and the box kernel; we make use of the latter as it can be shown to minimise the variance of \eqref{eq:isoxpcfest}:
\begin{equation*}
\kappa_{h_r}(s) = \left\{
\begin{array}{lcc}
1/2h_r  &\quad&  -h_r\le s \le h_r  \\
0 	  &\quad&  \textrm{otherwise.}
\end{array}
\right.
\end{equation*}

The second estimator that we detail here corresponds to the anisotropic sector-$K$-function $K^a_{pq}(r,\phi)$:
\begin{equation}
\label{eq:sectorKfunest}
\hat{K}^a_{pq}(r,\phi) = \hat{K}^a_{pq}(r,\phi+\pi) = \sum_{\substack{x_p\in X_p\cap W\\ x_q\in X_q\cap W}}^{\ne}\frac{H(x_p-x_q,(r,\phi))}{\hat{\lambda}_p \hat{\lambda}_q
	|W\cap W_{x_p-x_q}|},
\end{equation}
where
\begin{equation*}
H(x_1-x_2,(r,\phi)) = \mathbb{I}(\|x_1-x_2\| \le r)\kappa_{h_{\phi}}(\psi(x_1,x_2) - \phi),
\end{equation*}
with $\mathbb{I}(\cdot)$ the indicator function, $\kappa_{h_{\phi}}$ an angular kernel function with bandwidth $h_{\phi}$, and $\psi(x_1,x_2)$ the angle between the directed line from $x_1$ to $x_2$ and the abscissa-axis. In our implementation, we will use a box kernel for $\kappa_{h_{\phi}}$, defined analogously to the radial kernel function $\kappa_{h_r}$ above.

\subsection{Additional Figures}
\label{appx:addfigs}
In Section \ref{sec:implementation} of the article, we provide proof-of-concept results for our model-fitting procedure. There, we have given numerical summaries of the estimated parameter distributions for four distinct model specifications, along with an illustration, in Figure \ref{fig:MCsims_Model1}, corresponding to one of these models.

Here, we provide illustration of the estimated parameter distributions
for the three remaining model specifications in our proof-of-concept tests. Figures \ref{fig:MCsims_Model2}, \ref{fig:MCsims_Model3} and \ref{fig:MCsims_Model4} correspond to Models 2, 3 and 4, respectively, and the true parameter values used to generate each dataset can be found in Table \ref{table:MCsims_fullMatern}.
 
\begin{figure}
	\includegraphics[width=\textwidth]{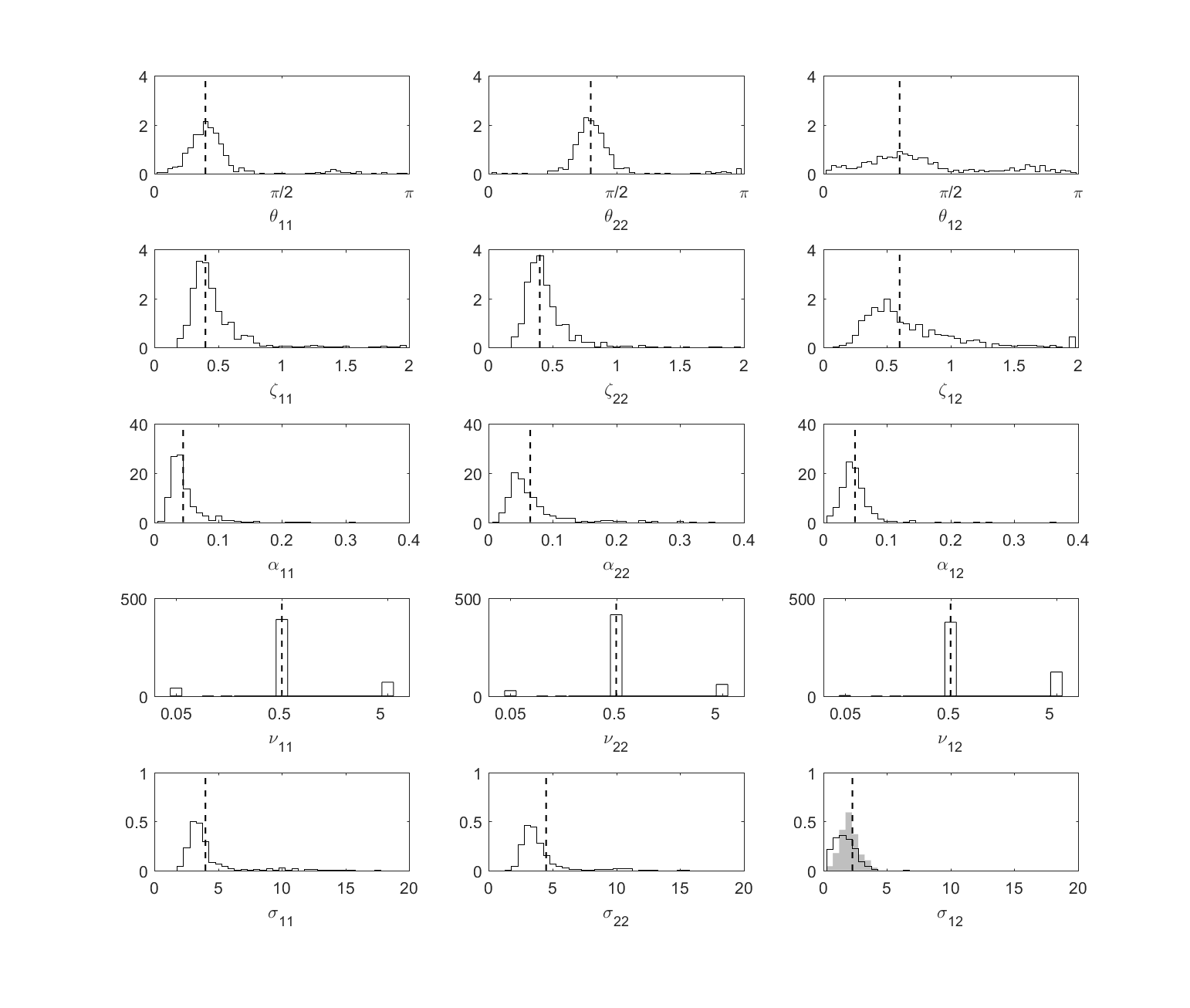}
	\vspace{-40pt}
	\caption{Histograms of the parameter distributions for the synthetic bivariate geometric anisotropic LGCP with Mat\'ern covariance structure specified by Model 2. The parameter values used to generate each dataset are marked by vertical dashed lines.}
	\label{fig:MCsims_Model2}
\end{figure}
\begin{figure}
	\includegraphics[width=\textwidth]{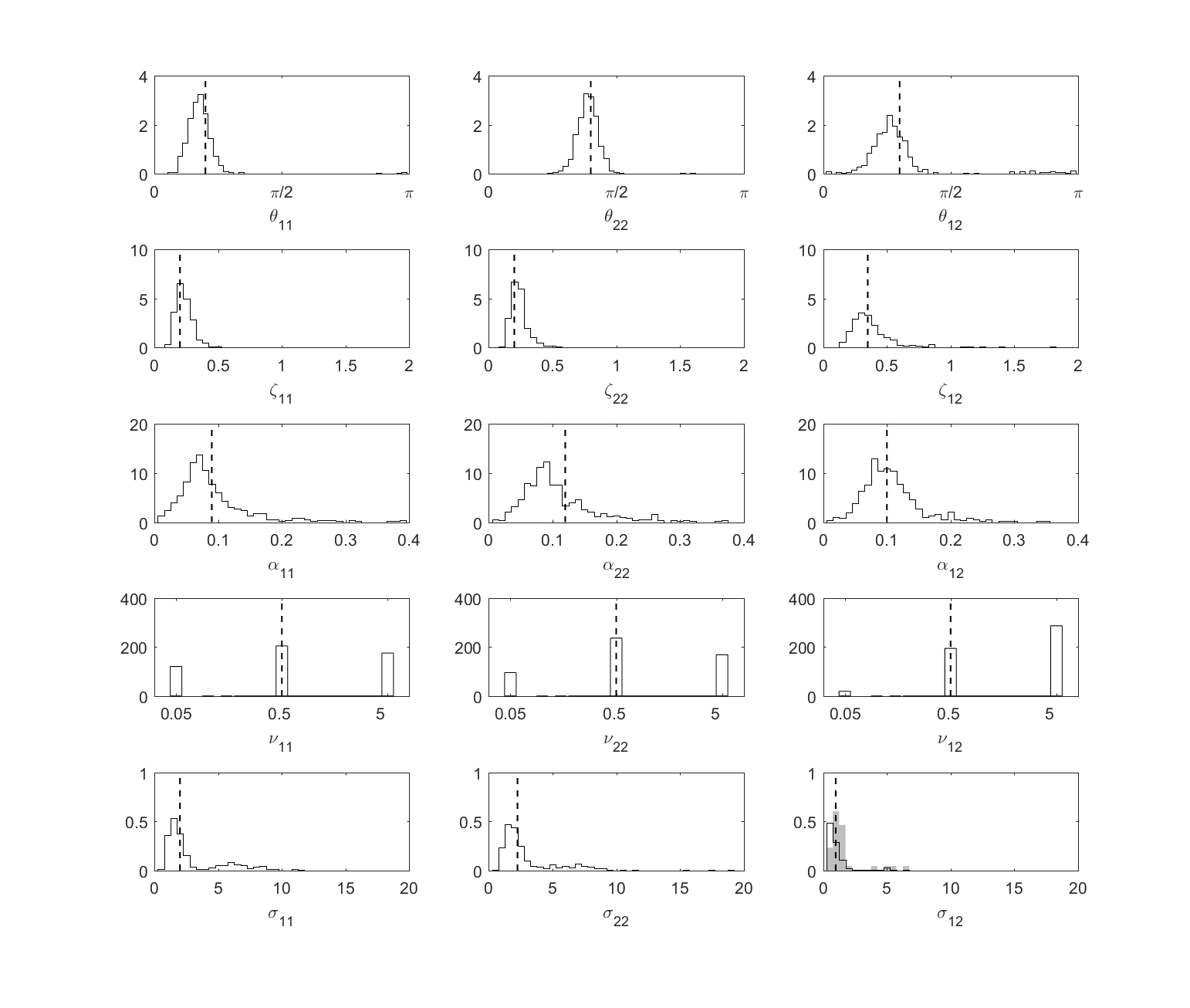}
	\vspace{-40pt}
	\caption{Histograms of the parameter distributions for the synthetic bivariate geometric anisotropic LGCP with Mat\'ern covariance structure specified by Model 3. The parameter values used to generate each dataset are marked by vertical dashed lines.}
	\label{fig:MCsims_Model3}
\end{figure}
\begin{figure}
	\includegraphics[width=\textwidth]{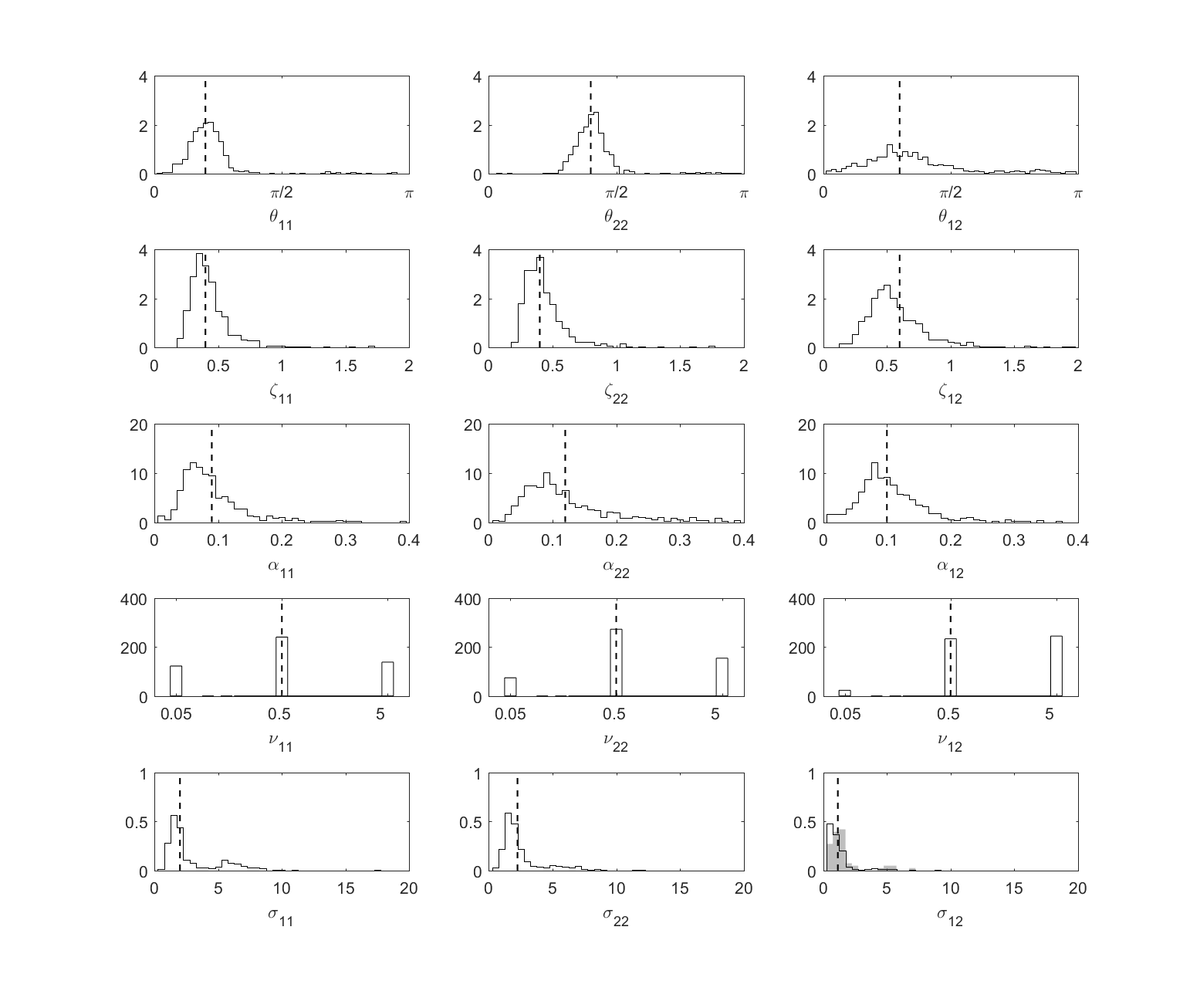}
	\vspace{-40pt}
	\caption{Histograms of the parameter distributions for the synthetic bivariate geometric anisotropic LGCP with Mat\'ern covariance structure specified by Model 4. The parameter values used to generate each dataset are marked by vertical dashed lines.}
	\label{fig:MCsims_Model4}
\end{figure}

\bibliographystyle{biometrika}
\bibliography{C://Users/jmartin/Documents/WrittenWork/FullBibliography}

\end{document}